\documentclass[11pt]{article}
\usepackage{graphicx}
\usepackage{latexsym}
\usepackage{amssymb}
\usepackage{amsmath}
\usepackage{amsthm}
\usepackage[margin=1in]{geometry}
\usepackage{paralist}
\usepackage{epstopdf}
\usepackage{hyperref}

\begin{document}

\newtheorem{theorem}{Theorem}[section]
\newtheorem{lemma}[theorem]{Lemma}
\newtheorem{corollary}[theorem]{Corollary}
\newtheorem{definition}{Definition}[section]
\newtheorem{proposition}[theorem]{Proposition}
\newtheorem{observation}[theorem]{Observation}
\newtheorem{example}{Example}

\newcounter{theremark}
\newenvironment{remark}{\refstepcounter{theremark}\medskip\textsc{Remark~\arabic{theremark}.}}{}

\newenvironment{repobservation}[1]{\noindent \sc{\bf Observation~\ref{#1}.}\begin{itshape}}{\end{itshape}}
\newenvironment{reptheorem}[1]{\noindent \sc{\bf Theorem~\ref{#1}.}\begin{itshape}}{\end{itshape}}
\newenvironment{replemma}[1]{\noindent \sc{\bf Lemma~\ref{#1}.}\begin{itshape}}{\end{itshape}}

\def \ee   {\varepsilon}
\def \PP   {{\cal P}}
\def\dist  {\mathsf{dist}}

\newcommand{\ie}{{\it i.e.}}
\newcommand{\eg}{{\it e.g.}}

\newcommand{\size}{\operatorname{size}}
\newcommand{\treesize}{\operatorname{treesize}}
\newcommand{\subtree}{\operatorname{subtree}}
\newcommand{\children}{\operatorname{children}}

\begin{titlepage}
\title{Mechanism Design with Strategic Mediators}

\author{Moshe Babaioff\thanks{Microsoft Research, Herzliya Israel, \texttt{moshe@microsoft.com}.}
\and
Moran Feldman\thanks{School of Computer and Communication Sciences, EPFL, \texttt{moran.feldman@epfl.ch}. Supported by ERC Starting Grant 335288-OptApprox. Part of the work was done while the author was an intern at Microsoft Research, Herzliya Israel.}
\and
Moshe Tennenholtz\thanks{Faculty of Industrial Engineering and Management, Technion, \texttt{moshet@ie.technion.ac.il}. Work carried out at Microsoft Research, Herzliya Israel.}
}

\maketitle

\begin{abstract}
We consider the problem of designing mechanisms that interact with strategic agents through {\em strategic} intermediaries (or mediators), and investigate the cost to society due to the mediators' strategic behavior. Selfish agents with private information are each associated with exactly one strategic mediator, and can interact with the mechanism exclusively through that mediator. Each mediator aims to optimize the combined utility of {\em his} agents, while the mechanism aims to optimize the combined utility of {\em all} agents. We focus on the problem of facility location on a metric induced by a publicly known tree. With non-strategic mediators, there is a dominant strategy mechanism that is optimal. We show that when both agents and mediators act strategically, there is no dominant strategy mechanism that achieves {\em any} approximation. We, thus, slightly relax the incentive constraints, and define the notion of a {\em two-sided incentive compatible} mechanism. We show that the $3$-competitive {\em deterministic} mechanism suggested by Procaccia and Tennenholtz~\cite{ProcacciaT13} and Dekel et al.~\cite{DekelFP2010} for lines extends naturally to trees, and is still $3$-competitive as well as two-sided incentive compatible. This is essentially the best possible~\cite{DekelFP2010,ProcacciaT13}.
We then show that by allowing randomization one can construct a $2$-competitive {\em randomized} mechanism that is two-sided incentive compatible, and this is also essentially tight.
This result also closes a gap left in the work of Procaccia and Tennenholtz~\cite{ProcacciaT13} and Lu et al.~\cite{LuWZ09} for the simpler problem of designing strategy-proof mechanisms for weighted agents with no mediators on a line, while extending to the more general model of trees.
We also investigate a further generalization of the above setting where there are multiple levels of mediators.
\end{abstract}

\pagenumbering{Alph}
\thispagestyle{empty}
\end{titlepage}

\pagenumbering{arabic}


\section{Introduction}

The Algorithmic Mechanism Design literature is generally interested in the implications of strategic behavior on the quality of social decision making.
The usual assumption is that agents interact directly with a mechanism that picks an outcome. Yet, in many complex real world settings the interaction goes through {\em intermediaries}. If these intermediaries are acting strategically, this can influence the outcome picked by the mechanism, and result with an increase in social cost.

Consider, for example, a political decision taken by indirect voting. There are districts, and each district is represented by a representative.
Each citizen has a position, and let us assume the positions of the citizens are points on an interval. A decision is also a point on the interval, and the cost for a citizen of such a point equals to the distance of her position from the decision made. Each representative aims to minimize the total cost for his own constituency, while the global goal is to minimize the total cost of all citizens. Decisions are taken using the reports of the representatives exclusively (there is no direct interaction with the agents), and these representatives have the freedom to manipulate their reports  if such a manipulation helps their constituency.
We are interested in questions such as: {\em What is the cost for society of such strategic behavior? How should the society set up the decision process to minimize that cost?}

More generally, we are interested in designing mechanisms that interact with strategic agents through {\em strategic} intermediaries (which we also call mediators).
Agents have private information, and when put in a game, each agent acts to optimize her own utility.\footnote{Throughout the paper we refer to an agent as ``she'', and to a mediator as ``he''.}
The mechanism designer aims to optimize a social goal.
The intermediaries do not have any private information of their own, rather, each intermediary acts in the mechanism on behalf of the agents associated with him, aiming to optimize the same social goal with respect to his agents only (note that he does {\em not} have a personal agenda and is completely benevolent).
As the intermediaries control the information flow from the agents to the mechanism, the mechanism faces strategic behavior not only of the agents, but also of intermediaries:
within the freedom given by the mechanism, an intermediary acts strategically to optimize on behalf of the agents he represents.\footnote{We assume that an intermediary {\em is able} to manipulate the reports of his agents, and do not consider settings in which there exists an infrastructure for sending messages between the agents and the center through the mediators in a non-manipulable way (\eg, using cryptographic means.)}
In this paper we aim to understand the implications of the strategic behavior of intermediaries on the welfare of the agents.

The general framework outlined above can be studied in the context of many specific settings, and might yield very different results in different cases. Here, we focus on one such example and leave the consideration of other settings for future works. The setting we consider is facility location on a metric induced by a publicly known tree, which generalizes the decision making problem on a line introduced above.
There are $n$ agents, each located at some private location. The agents are partitioned to $k$ disjoint sets, and each set is represented by a unique mediator. The mechanism (or {\em center}) should locate one facility. The cost of an agent is her distance from the location of the facility, and she aims to minimize her cost.\footnote{While the general framework does not preclude transfer of utilities, in this specific model there is no money and utilities cannot be transferred. Thus, our results for facility location can be viewed as part of the literature on approximate mechanism design without money \cite{ProcacciaT13}.}
The social goal considered is the goal of minimizing the total distance of the agents from the facility.

If the center had access to the locations of all agents he could minimize the total cost by locating the facility at a {\em median} of all locations.  While all our results hold for general tree metrics, for the sake of the exposition, in the introduction we mainly discuss the euclidian metric on an interval of the real line. For that metric, if $t_i$ is the $i$-th left most agent (breaking ties arbitrarily) and $n$ is odd, then there is a unique optimal location at the median location $t_{(n+1)/2}$ (for even $n$ there is an interval of optimal locations, between the two medians).
With strategic agents but no mediators (or equivalently, with non-strategic mediators), there is a dominant strategy mechanism that is optimal: locate the facility at a median point, breaking ties to the left.
While this result gives a complete picture for the standard model without strategic mediators, we show that with both strategic agents and strategic mediators the picture is much more complicated. We first show that there does not exist a dominant strategy mechanism achieving any approximation. This happens even in a simple setting with two possible locations, a single mediator and a single agent, as if the agent switches between the locations in her report, the mediator should switch them back, and vice versa.

Given the impossibility to achieve a dominant strategy mechanism with good performance, we suggest a slightly weaker solution concept for direct revelation mechanisms (in which each agent reports her location, and each mediator reports the locations of all his agents).
Our aim would be to build mechanisms which achieve good approximation (minimize the ratio between the cost of the outcome and the optimal cost).
A mechanism is {\em agent-side incentive compatible (agent-side IC)} if each agent has a dominant strategy to be truthful given that her mediator is truthful (regardless of any parameter of the model, like the number of mediators, and regardless of other players' strategies).
A mechanism is {\em mediator-side incentive compatible (mediator-side IC)} if each mediator has a dominant strategy to be truthful given that all his agents are truthful
(again, regardless of any parameter and regardless of other players' strategies).
We aim to construct mechanisms that are {\em two-sided incentive compatible (two-sided IC)}, \ie, they are both agent-side incentive compatible and mediator-side incentive compatible. We construct both deterministic and randomized mechanisms, and prove that they achieve essentially the best possible performance.\footnote{All the mechanisms we construct run in polynomial time, while our lower bounds hold independent of computational consideration. Like in prior literature in approximate mechanism design without money \cite{ProcacciaT13}, the barrier to optimality is incentives, not computation.}

One of the settings considered by Procaccia and Tennenholtz~\cite{ProcacciaT13} is equivalent to designing deterministic mediator-side IC mechanisms on an interval of the real line. Their work implies that the results of Dekel et al.~\cite{DekelFP2010} for regression learning induce a $3$-competitive \emph{deterministic} mediator-side IC mechanism on an interval, and that this is \emph{essentially} the best possible competitive ratio for such a mechanism. The mechanism induced works as follows: for every mediator, it replaces all points reported by the mediator by the optimal\footnote{There might be multiple optimal locations, in such cases ties need to be handled carefully to preserve incentives. To simplify the exposition, in the introduction we assume there are no ties.} location for that mediator, and then finds an optimal location with this new input (the mechanism essentially computes median of medians, weighted by the number of agents each mediator represents).

We prove the above mechanism is also agent-side IC, and describe a simple extension of it to general trees. This yields the following theorem.

\begin{theorem}
There exists a deterministic two-sided IC mechanism on \emph{tree metrics} with a competitive ratio of $3$.
Moreover, for any fixed $\ee > 0$, there is no deterministic two-sided IC mechanism with a competitive ratio of $3-\ee$.
\end{theorem}

Procaccia and Tennenholtz~\cite{ProcacciaT13} raised the question whether it is possible to get a better competitive ratio using randomization. They were able to answer affirmatively in the case of two mediators representing a ``similar'' number of agents. Lu et al.~\cite{LuWZ09} extend the analysis of the mechanism of \cite{ProcacciaT13} to the case of multiple mediators representing a ``similar'' number of agents. However, even if all mediators have equal number of agents, the competitive ratio of this mechanism approaches $3$ as the number of mediators increase. On the negative side, \cite{LuWZ09} gives a hardness result of $1.33$ using a complex LP-based proof.

We suggest a new and sophisticated randomized mechanism that is $2$-competitive and works for any tree. We also prove using a simple argument that this is essentially the best possible.\footnote{Like the hardness of \cite{LuWZ09}, our hardness holds, in fact, even for mediator-side IC mechanisms.}
\begin{theorem}
There exists a randomized two-sided IC mechanism on tree metrics with a competitive ratio of $2$.
Moreover, for any fixed $\ee > 0$, there is no randomized two-sided IC mechanism with a competitive ratio of $2-\ee$.
\end{theorem}
This result closes the gap left in the work of Procaccia and Tennenholtz~\cite{ProcacciaT13} and Lu et al.~\cite{LuWZ09} for the simpler problem of designing strategy-proof mechanisms for weighted agents with no mediators on a line, while extending to the more general model of trees.

For the case of locations on an interval of the real line the mechanism works as follows. For every mediator it replaces all points reported to the mediator by the optimal location for that mediator. For simplicity assume that the number of agents can be divided by $4$. Then, it sorts the locations and uses a uniformly selected point among the $n/2$ central points (that is, the points from the $n/4 + 1$ leftmost location to the $3n/4$ leftmost location).

The randomized mechanism for trees generalizes this idea but is much more involved, and is our main technical contribution.
This mechanism chooses from the set of medians (optimal locations of mediators) a ``core'' subset. This core is the equivalent of the central points from the line case. Each point in the core is assigned some positive probability to become the facility location. However, unlike in the line case, the probabilities assigned to the points of the core are non-uniform, and are carefully chosen to achieve both the competitive ratio and the right incentives. The exact probability distribution depends on the medians of all mediators, including medians outside of the core. If all the reports happen to fall on a single line, then the probability distribution becomes uniform, and the algorithm reduces to the one described above for lines.

We remark that all our mechanisms run algorithms that use only the optimal location for each mediator,
and do not need, in addition, access to the exact locations of the agents associated with each mediator.
We call an algorithm that satisfies this property a {\em mediator based} algorithm.
We prove that such algorithms, which use {\em only} the locations of the optimal points of the mediators (and not the locations of their agents),
cannot be better than $2$-competitive.
Interestingly, we show that there exists a deterministic mediator based algorithm that has a competitive ratio of $2$, yet that algorithm is {\em not} two-sided IC. Thus, for deterministic two-sided IC mechanisms, the implications of strategic behavior by the mediators goes beyond the constraint of being mediator based; such mechanisms cannot be better than $3$-competitive (which is tight). Thus, there is a gap that is a result of incentives, and is not due to insufficient information.

Tree metrics are a strict generalization of line metrics and capture domains that cannot be reasonably modeled by line metrics. Consider the following toy example. People of three nationalities live in a single country (\eg, Switzerland), and want to elect a president. The candidates for the position differ in two attributes: their nationality and their degree of nationalism (for example, how much are they willing to settle for a compromise when dealing with an issue on which the different national groups disagree). Each citizen, naturally, wants to elect a president sharing his nationality, but different citizens of the same national group might want to elect candidates with different degrees of nationalism. Notice that a candidate of low nationalism is more acceptable by citizens of other nationalities (regardless of the level of nationalism, every citizen would probably like to have a president of her own nationality), thus, the metric induced by this example is a star with 3 edges (of course, one can think of a country with more nationalities to get a star with more edges).

We also consider a generalization of the above setting allowing multiple levels of mediation. In other words, the center, agents and mediators form a tree, in which the root is the center and the leaves are the agents. Every internal node of the tree is a mediator representing its children in the tree. Unfortunately, the competitive ratio of every mechanism for this setting degrades exponentially with the height of the tree, even when the mechanism is only required to respect a very weak definition of incentive compatibility. 
This result is consistent with the existence of a symmetric voting system composed of $k$ levels where a minority of size exponentially small in $k$ can control the decisions of the system.\footnote{See \texttt{http://gowers.wordpress.com/2013/10/15/holding-a-country-to-ransom/} for an example of such a voting system.} Finally, we show that the mechanism that iteratively applies weighted median has a competitive ratio which is essentially optimal and satisfies the weak notion of incentive compatibility.

\subsection{Related Work}

In this paper we deal with mediators who act as intermediaries between a set of agents and a mechanism. The most related setting studied in the literature is the recent work on auctions with intermediaries \cite{FeldmanMMP10}. There, as in our setting, both agents and intermediaries are strategic. However, the setting there is Bayesian while ours is Pre-Bayesian. Also, our aim is to address the social welfare issue requiring dominant strategies by the agents when their associated mediator is truthful, rather than revenue maximization.

More generally, our work refers to the study of mediators (see, \eg, \cite{MondererT09} for a study in the context of complete information games, and \cite{AshlagiMT09} for a study in the context of incomplete information games).
However, the typical signature of work on mediators is a single mediator that serves as an arbitration device: the agents are not a captive audience, and each of them can decide to participate in the game directly or work through the mediator. In our setting there are multiple intermediaries, each having his own captive audience, which must play the game through the intermediary. Moreover, the intermediaries are players and try to optimize their own utilities.
Our setting nicely fit with situations such as voting by the (already selected) representatives of a geographic area or interest group. Additionally, we would like to mention the work of Leyton-Brown et al.~\cite{Leyton-BrownST02} which deals with game theoretic aspects of bidding clubs in which ``collusion devices'' (cartels) are strategically created in a fixed mechanism (first price auction). In contrast, in our setting the partition of agents to mediators is pre-determined and our focus is on mechanism design given that fact.

The specific example of the framework that we consider is related to the recent literature on approximated mechanism design without money \cite{ProcacciaT13}. This literature deals with approximation algorithms which are used to resolve incentive issues for tractable problems rather than overcome computational complexity of intractable problems, when no money transfers are available. An additional conceptual contribution of our work is extending the literature on approximate mechanism design without money to incorporate mediation devices. Indeed, the problem studied in this paper, the facility location problem, is the canonical problem of that literature, which is easily solved (optimally) if no intermediaries are in place.

As pointed out in the previous section, the design of mediator-side incentive compatible mechanisms is equivalent to the design of strategy-proof mechanisms for weighted agents that was studied by Procaccia and Tennenholtz~\cite{ProcacciaT13}, and later also by Lu et al.~\cite{LuWZ09} (these papers only considered the special case of a line metric). The implications of this equivalence to our settings were discussed above.

The literature on facility location on a line based on information provided by strategic agents  is classic in the context of mechanism design with single-peaked preferences \cite{Moulin80}. The extension of this problem to facility location on a network has been introduced by \cite{SchummerV02}. It has been shown that there exist non-dictatorial strategy proof algorithms for facility location on trees, and that any graph possessing circles does not allow for that. The study of approximate mechanism design without money for networks \cite{AlonFPT10} discusses the minimization of the maximal distance to a facility on a network using deterministic and probabilistic strategy proof algorithms, yielding some positive approximation results and tight bounds. The problem of approximating the optimal location of two facilities on a line using strategy proof mechanisms has been discussed in \cite{LuSWZ10}, while the general case of locating $k$ facilities in an approximately optimal manner using strategy proof mechanisms can be handled for large populations by the general technique given
in \cite{NissimST12}.

\section{Model and Solution Concept} \label{sec:model}

Within the general framework of strategic mediators we focus on one specific mechanism design problem: facility location on a metric induced by a publicly known tree $T = (V, E)$ with the following metric on each edge.
Each edge $e\in E$ in the tree is mapped to the interval $[0,\ell_e]$ for some $\ell_e>0$, with the usual Euclidian metric.
In our problem there are $n$ agents, each of which has a private position which can be represented by a point on the tree. The position of an agent can be either a node $v \in V$ or a point somewhere along an edge $e \in E$. Each one of the $n$ agents is associated with one of $k$ mediators.
For $i\in [k]$, mediator $d_i$ represents a set $A_i$ of $n_i$ agents; we denote these agents by $a_{i, 1}, a_{i, 2}, \ldots, a_{i, n_i}$. As we assume each agent is associated with exactly one mediator, the sets of agents of any two mediators do not intersect and  $\sum_{i = 1}^k n_i = n$.
The position of each agent is only known to the agent herself, and we denote the private position of agent $a_{i,j}$ by $t_{i,j}$.
Everything else is common knowledge. In particular, the number of agents represented by each mediator is known to the mechanism.\footnote{If the number of agents represented by each mediator is considered private information, then the mechanism has no way to distinguish ``important'' mediators representing many agents from ``unimportant'' mediators representing only few agents. This intuitive impossibility can be easily formalized to show that no constant competitive ratio is possible. The assumption that size of the population represented by each mediator is public is reasonable in many settings, for example, the size of the population of a congressman's district is publicly known.}
We call a point from the metric induced by $T$ simply a ``point''. For example, by saying that ``$p$ is a point'' we mean that $p$ is a point from the metric induced by $T$. Particularly, the location of each agent is a point.

The {\em center} has to pick a position for a single facility.
If the center locates the facility at point $p$, then the {\em cost of an agent} $a_{i,j}$ is $\dist(p, t_{i,j})$, where $\dist(p, t_{i,j})$ is the distance between $p$ and $t_{i,j}$ along the metric induced by $T$.
The \emph{social cost} of locating the facility at point $p$ is $\sum_{i\in[k],j\in[n_i]} \dist(p, t_{i,j})$, \ie, the sum of all the agents' costs.
The objective of the center is to pick a location for the facility that minimizes the social cost.
The {\em cost of a mediator} $d_i$ ($i\in[k]$) is the total cost for the agents he represents, which is $\sum_{j\in[n_i]} \dist(p, t_{i,j})$.
Each mediator aims to minimize his cost.
We use the term {\em player} to denote either an agent or a mediator.
We assume that the center and players are risk neutral, and for a distribution over locations, they evaluate their cost by the expected cost.
Note that in our model there is no money, and utilities cannot be transferred.

An {\em algorithm} for the center is a mapping from its input, the locations of all agents, to a location for the facility.
We say that an algorithm is $\alpha$-competitive, or has a \emph{competitive ratio} of $\alpha$, if for any set of locations for the agents, the location picked by the center for the facility induces a cost that is at most $\alpha$ times larger than the minimal possible cost (with respect to its input).

When the agents' locations are private information the center has to come up with a mechanism by which players report their information and this information is used to pick a facility location.
We consider {\em direct revelation mechanisms} in which each agent is asked to report her location (to her mediator), and each mediator is asked to report the location of each of his agents. The mechanism uses the public information and the mediators' reports to locate the facility, with the aim of minimizing the social cost. We say that a mechanism is {\em $\alpha$-competitive}, or has a {\em competitive ratio} of $\alpha$, if under the solution concept that we consider, the location of the facility picked by the center has cost that is at most $\alpha$ times larger than the minimal possible cost.
Crucial to our model is the assumption that the center (or the mechanism) can not interact directly with the agents, and has access to their locations only through their mediators, which can manipulate the agents' reports.

\paragraph{Solution Concept.} Any direct revelation mechanism picked by the center puts the agents and the mediators (the players), which are both strategic, into a game. We would like to use mechanisms which induce games with some desired properties.

A direct revelation mechanism is {\em dominant strategy truthful} if it is a dominant strategy for each agent to report her location truthfully (regardless of the strategies chosen by the mediators and the other agents), and it is a dominant strategy for each mediator to report the locations of all of his agents to the center exactly as reported to him by the agents (again, regardless of the strategies chosen by the agents and the other mediators). We observe that asking a competitive mechanism to be dominant strategy truthful is unrealistic.

\begin{observation} \label{ob:no_dominant}
No direct revelation dominant strategy truthful mechanism has a finite competitiveness, and this is true even if the center is allowed to charge the mediators and agents.
\end{observation}
\begin{proof}
Consider an instance with a single mediator representing a single agent which can take two possible locations $x$ and $y$.
To have finite competitiveness the center must locate the facility at the location of the agent (when both the agent and the mediator are truthful). However, the center gets no information other than the report of the mediator, and therefore, it must \emph{always} locate the facility at the location reported by the mediator. Moreover, the charges collected by the center can depend only on this location.

Let $p_x$ and $p_y$ be the charges that the mediator pays when reporting $x$ and $y$, respectively. Assume without loss of generality that $p_y \geq p_x$. Now, assume that the agent's strategy is to report $y$ despite the fact that she is located at $x$. If the mediator switches the location back, then his cost is $p_x$, while a truthful repetition of the agent's report will result in a cost of $p_y + |y - x| \geq p_x + |y - x| > p_x$. Thus, it is clearly non-optimal for the mediator in this case to truthfully repeat the report of the agent.
\end{proof}

\begin{remark}
The above impossibility applies to a setting in which all entities have {\em exactly} the same utility function, so there are no conflicts. It is a result of the sequential nature of information propagation from the agents to the center through the mediators, and the incompatibility of that with dominant strategies.
\end{remark}

Given this impossibility result we need to settle for a slightly weaker solution concept, achieving \emph{Incentive Compatibility} (IC) in the following sense.
We still want each agent to have an incentive to be truthful, as long as her mediator is truthful (as opposed to playing an ``unreasonable'' strategy), and we want each mediator to be truthful as long as his agents are truthful. This is captured by the following definition.

\begin{definition}
A direct revelation mechanism is {\em agent-side incentive compatible} if for every mediator $d_i$, in the induced game created by fixing $d_i$ to be truthful, truthful reporting is a dominant strategy for each agent $a_{i, j}$ represented by $d_i$.

A direct revelation mechanism is {\em mediator-side incentive compatible} if for every mediator $d_i$, in the induced game created by fixing all $d_i$'s agents to be truthful, truthful reporting is a dominant strategy for $d_i$.

A direct revelation mechanism is {\em two-sided incentive compatible} if it is agent-side incentive compatible and mediator-side incentive compatible.
\end{definition}
Note that in any two-sided incentive compatible mechanism, it is in particular an Ex-post Nash for all players to be truthful.

To understand the implication of strategic behavior by the mediators we compare the competitiveness achieved by the best two-sided incentive compatible mechanisms to the competitiveness achieved by the best agent-side incentive compatible mechanisms (we do so both for deterministic and for randomized mechanisms).

\subsection{Agent-Side Incentive Compatible Mechanisms}
\label{sec:optimal}
Median points play a significant role both in the optimal algorithm and our mechanisms. We next present some basic definitions and observe that median points exactly characterize optimal locations.

\begin{definition}
Following are the definitions of median points and weighted median points.
\begin{compactitem}
\item A \emph{weighted point} is a pair $(p, x)$ where $p$ is a point and $x$ is a positive real number. Given a weighted point $\tilde{p} = (p, x)$, we say that $x$ is the weight of $\tilde{p}$, and write $w(\tilde{p}) = x$. We also think of $\tilde{p}$ as located at location $p$ in the metric. Hence, we can talk, \eg, about the distance between two weighted points.\footnote{In the interest of readability, throughout the paper we put the tilde sign above letters representing weighted points.}
\item Given a multi-set $S$ of elements, let $f(p, S)$ be the multiplicity of $p$ in $S$ (\ie, $f(p, S)$ is the number of copies of $p$ in $S$). Given an additional multi-set $S'$, we denote by $S \cup S'$ and $S \setminus S'$ two multi-set containing $f(p, S) + f(p, S')$ and $\max\{f(p, S) - f(p, S'), 0\}$ copies of every element $p \in S\cup S'$, respectively.
\item Given a multi-set $S$ of weighted points and a point $p$.
    \begin{compactitem}
        \item Let $S_p$ denote the multi-set of weighted points in $S$ that have $p$ as their location. More formally, for every weighted point $\tilde{q} = (q, x) \in S$, the multi-set $S_p$ contains $f(\tilde{q}, S)$ copies of $\tilde{q}$ if $q = p$, and no copy of $\tilde{q}$ otherwise.
        \item The weight of $S$, denoted by $w(S)$, is the total weight of the weighted points in $S$. More formally, $w(S) = \sum_{\tilde{q} \in S} w(\tilde{q}) \cdot f(\tilde{q}, S)$,
        \item Let $m_p$ be the maximum weight of a multi-set $S' \subseteq S \setminus S_p$ such that the path connecting every two weighted points of $S'$ does not go through $p$.
        \item We say that $p$ is a {\em weighted median} of $S$ if $m_p \leq w(S) / 2$.
    \end{compactitem}
\item Given a multi-set $S$ of points, we say that a point $p$ is a {\em median} of $S$ if it is a weighted median of the multi-set $S'$ containing $f(p, S)$ copies of $(p, 1)$ for every point $p \in S$.
\end{compactitem}
\end{definition}

Informally, a point $p$ is a median of $S$ if removing it splits the tree $T$ into parts, each containing at most $|S|/2$ points of $S$. The importance of median points stems from the following easy observation, whose proof (with simpler versions going back to~\cite{Jordan1869})
is deferred to Appendix~\ref{ap:missing_model_proofs}.

\begin{observation} \label{ob:median_optimal}
For every non empty finite multi-set $S$ of weighted points, there is at least one weighted median. Moreover, a point $p$ is a weighted median of $S$ if and only if locating the facility at $p$ minimizes the weighted total cost of a set of agents located at the points of $S$ (\ie, the sum $\sum_{\tilde{q} \in S} w(\tilde{q}) \cdot f(\tilde{q}, S) \cdot \dist(p, \tilde{q})$).
\end{observation}

The mechanism that always picks a median of all the locations of the agents (with a careful tie breaking) is optimal and agent-side incentive compatible.
Thus, we have the following observation that naturally extends a well known result for line metrics \cite{ProcacciaT13}. We present its proof for completeness in Appendix~\ref{ap:missing_model_proofs}.

\begin{observation} \label{ob:agent_side_optimal}
There exists a deterministic agent-side incentive compatible mechanism that is optimal (\ie, $1$-competitive).
\end{observation}

We note that the optimal agent-side IC mechanism is deterministic. Hence, randomization clearly does not help in improving performance when mediators are not strategic. Our results show that this is not the case when mediators are strategic and one aims for two-sided IC mechanisms.

\subsection{Mediator Based Algorithms}

We say that an algorithm for the center is {\em mediator based} if it uses only an optimal facility location for each mediator (but never uses any other information regarding the positions of the agents themselves).
We show that for mediator based algorithms, randomization does not improve performance, as any such randomized $\alpha$-competitive algorithm can be transformed to a mediator based deterministic algorithm with the same competitive ratio (moreover, the resulting algorithm performs at least as good on every single input).
To state this result we first need the following lemma whose proof appears in Appendix~\ref{ap:missing_model_proofs}.

\begin{lemma} \label{le:single_better_than_distribution}
For any tree and any distribution over points $F$, there exists a point $p(F)$ such that for any finite multi-set $S$ of points:
\[
    \underset{p' \sim F}{\mathbb{E}} \left[\sum_{q \in S} \dist(p', q) \cdot f(q, S)\right]
    \geq
    \sum_{q \in S} \dist(p(F), q) \cdot f(q, S)
    \enspace.
\]
Moreover, for the euclidian metric on $[a, b]$ (for arbitrary $a$ and $b$) the expected location according to $F$ can serve as such a point $p(F)$.
\end{lemma}
Note that $p(F)$ does {\em not} depend on $S$, and the same $p(F)$ works for every $S$.
A randomized algorithm maps the locations of all agents to a distribution over locations $F$. By the above lemma the deterministic algorithm that instead locates the facility deterministically at $p(F)$ can only improve the social cost for every input. Thus, we have the following corollary.
\begin{corollary}
Given any $\alpha$-competitive mediator based {\em randomized} algorithm, it is possible to construct a mediator based {\em deterministic} algorithm with the same competitive ratio $\alpha$.
\end{corollary}

Note that the above transformation does {\em not} maintain incentives, indeed we show below that there exists a mediator based randomized two-sided IC mechanism which is $2$-competitive, but no mediator based deterministic two-sided IC mechanism achieves this competitive ratio. Thus, although randomization does not improve performance for mediator-based {\em algorithms}, it does  improve performance for mediator-based {\em two-sided IC mechanisms}.

\section{Deterministic Two-Sided IC Mechanisms} \label{sc:deterministic_algorithm}

In this section we extend the deterministic ``median of medians'' mechanism of~\cite{DekelFP2010,ProcacciaT13} from lines to trees and show that
the resulting mechanism, which we call the \emph{Weighted Median Mechanism} (WMM), is a deterministic two-sided IC mechanism. This mechanism is also $3$-competitive, which is essentially tight by a lower bound of~\cite{DekelFP2010,ProcacciaT13} (given for completeness as Theorem~\ref{th:lower-bound-det-mech}).
The mechanism essentially elicits from each mediator an optimal location from the mediator's perspective (median of the mediator's agents),
and then picks a weighted median of these locations. To create the right incentives for the agents and mediators, tie breaking must be handled carefully in both steps of the mechanism.
By breaking ties in a way that is independent of the players' reports, we make sure the players have no incentive to manipulate.
The basic idea is that in each step we break ties in favor of the point closest to an arbitrary predetermined point.
To formally describe WMM we need the following observation which proves that the above tie breaking rule is well defined, \ie, whenever the mechanism has to decide between a set of points, there is always a unique point in the set which is closest to the arbitrary predetermined point.

\begin{observation} \label{ob:closest_median_exists}
Given a non-empty finite multi-set $S$ of weighted points and an arbitrary point $z$, the set $M$ of weighted medians of $S$ contains a unique point $p$ closest to $z$.
\end{observation}
\begin{proof}
The proof of this observation consists of several steps. First, showing that it is enough to consider a finite set $R$ of weighted medians because for every weighted median $p$ either $p \in R$ or there is a weighted median $r(p) \in R$ which is closer to $z$ than $p$. Second, restricting ourselves to the set $R'$ of weighted medians from $R$ that are closest to $z$ (\ie, as close to $z$ as any other weighted median of $R$). Finally, proving that $R'$ is of size $1$, which completes the proof of the theorem.

Given a point $p \in M$, let us denote by $r(p)$ the first point of $V \cup \{z\} \cup \{p' ~|~ \exists_x (p', x) \in S\}$ along the path from $p$ to $z$ (there is such a point because $z$ is always on this path and the set is finite). Let us prove that $r(p)$ is also a weighted median of $S$ (\ie, $r(p) \in M$). Consider an arbitrary sub multi-set $S'$ of weighted points from $S$ which are connected via paths that do not go through $r(p)$. If $p \in V \cup \{z\} \cup \{p' ~|~ \exists_x (p', x) \in S\}$, then $r(p) = p$, and therefore, the weighted points of $S'$ are connected via paths that do not go through $p$ as well. Otherwise, every path connecting two weighted points of $S$ which goes through $p$ must go through $r(p)$ as well. Hence, once again, the weighted points of $S'$ are connected via paths that do not go through $p$. In both cases, we get $w(S') \leq w(S) / 2$ because $p$ is a weighted median of $S$, which completes the proof that $r(p) \in M$.

Let $R = \{r(p) ~|~ p \in M\}$ and $R' = \{p \in R ~|~ \dist(p, z) = \min_{q \in R} \dist(q, z)\}$ (the minimum in the definition of $R'$ always exists because $R$ is a finite set). Observe that $R' \subseteq R \subseteq M$, and for every point $p \in M$ either $p = r(p)$ or $\dist(p, z) > \dist(r(p), z)$. Hence, every point in $R'$ is closer to $z$ than any point of $M \setminus R$. Moreover, all the points in $R'$ have the same distance from $z$. To complete the proof we have to show that $|R'| = 1$. Clearly, $R' \neq \varnothing$, so it is enough to show $|R'| \leq 1$. Assume towards a contradiction that there are two points $p\neq q$ in $R'$. Let $y$ be the first point where the paths from $p$ and $q$ to $z$ meet. Since $\dist(p, z) = \dist(q, z)$, $y$ cannot be either $p$ or $q$, and therefore, $\dist(y, z) < \dist(p, z) = \dist(q, z)$. Observe that $y$ is located along the path connecting $p$ and $q$, and therefore, we can partition $S$ into three sub multi-sets whose union is $S$:
\begin{compactitem}
    \item $S_1$ - Contains weighted points of $S$ for which the path connecting them to $y$ goes through $p$.
    \item $S_2$ - Contains weighted points of $S$ for which the path connecting them to $y$ goes through $q$.
    \item $S_3$ - The rest of the weighted points of $S$ (the path connecting them to $y$ does goes through neither $p$ nor $q$).
\end{compactitem}
All the weighted points of $S_1 \cup S_3$ are connected by paths that do not go through $q$, and therefore, $w(S_1 \cup S_3) \leq w(S) / 2$ because $q$ is a weighted median of $S$. Thus, $w(S_2) = w(S) - w(S_1 \cup S_3) \geq w(S) / 2$. Using a similar argument we can also get $w(S_1) \geq w(S) / 2$. As $S_1 \cup S_2 \subseteq S$, the two results can hold at the same time only if $w(S_1) = w(S_2) = w(S) / 2$ and $w(S_3) = 0$.

Consider now an arbitrary sub multi-set $S'$ of weighted points of $S$ that are connected by paths that do not go through $y$. Clearly, either $S' \subseteq S_1 \cup S_3$ or $S' \subseteq S_2 \cup S_3$. In the first case $w(S') \leq w(S_1 \cup S_3) = w(S) / 2$, and in the second case $w(S') \leq w(S_2 \cup S_3) = w(S) / 2$. Hence, in both cases $w(S') \leq w(S) / 2$, which implies that $y \in M$. However, this is a contradiction since $\dist(y, z) < \dist(p, z) = \dist(q, z)$.
\end{proof}

The {\em Weighted Median Mechanism (WMM)} is a direct revelation mechanism in which the center does the following:
\begin{compactitem}
    \item For each mediator $d_i$ it computes $\ell_i$ which is the median of the multi-set $\{t'_{i, j} | 1 \leq j \leq n_i\}$ closest to $z_i$, where $t'_{i, j}$ is the location reported by $d_i$ for agent $a_{i, j}$ and $z_i$ is an arbitrary point chosen independently of the reports received from the mediators (such a median exists, and is unique, by Observation~\ref{ob:closest_median_exists}).
    \item Let $M$ be the set of weighted medians of the multi-set $\{(\ell_i, n_i) | 1 \leq i \leq k\}$.
    Locate the facility at the point of $M$ closest to $z$, where $z$ is an arbitrary point chosen independently of the reports received from the mediators.
\end{compactitem}

Note that this direct revelation mechanism can also be executed with much less communication since the only information the center needs from each mediator is a single point (the location of the median closest to some arbitrary point), and not the location of every agent represented by the mediator. Thus, the center can ask each mediator $d_i$ to report a single location $\ell_i$, and locate the facility at the weighted median of the multi-set $\{(\ell_i, n_i) | 1 \leq i \leq k\}$ closest to some point $z$ picked in advance. Observe that this algorithm for the center is mediator based. The resulting mechanism clearly achieves the same competitiveness as the direct revelation mechanism when each mediator $d_i$ indeed reports a median location $\ell_i$ of his agents closets to an arbitrary point $z_i$ since the location is picked using exactly the same method. Moreover, this mechanisms is also two-sided IC since the space of possible deviations for the mediators in this mechanism is more restricted than the corresponding space in the direct revelation mechanism.

On line metrics WMM is essentially identical to a mechanism already known to be $3$-competitive and mediator-side IC by an observation of \cite{ProcacciaT13} based on a result of \cite{DekelFP2010}. We present a proof that WMM is $3$-competitive and mediator-side IC not only for lines, but also for trees. We also show that it is agent-side IC. The next theorem summarizes the proved properties of WMM.

\begin{theorem} \label{thm:wmm_properties}
For any tree metric, the Weighted Median Mechanism is a deterministic two-sided IC mechanism with a competitive ratio of $3$.
\end{theorem}

The proof of Theorem~\ref{thm:wmm_properties} can be found in Section~\ref{ssc:wmm}. The following theorem proved by~\cite{DekelFP2010, ProcacciaT13} shows that WMM has an optimal competitive ratio. For completeness, we give a proof of this result in Section~\ref{sec:lower}.

\begin{theorem}
\label{th:lower-bound-det-mech}
Fix any constant $\ee > 0$.
Then, no direct revelation \emph{deterministic} mechanism that is mediator-side incentive compatible has a competitive ratio of $3-\ee$, even for line metrics.
\end{theorem}

\subsection{Analysis of WMM} \label{ssc:wmm}

In this section we analyze WMM and prove Theorem~\ref{thm:wmm_properties}. We begin with the following  auxiliary observation. This observation shows that, given a fixed arbitrary point $z$, if the weighted median of a multi-set $S$ closest to $z$ changes when some weight moves from the point $p$ to the point $p'$, then this weighted median must be located along the path connecting $p$ and $p'$, and moves along this path in the direction from $p$ to $p'$.
\begin{observation} \label{ob:point_moves}
Consider:
\begin{compactitem}
    \item Two multi-sets of weighted points $S$ and $S'$ differing in a single weighted point, \ie, $S' = S \setminus \{\tilde{p}\} \cup \{\tilde{p}'\}$ for a pair of weighted points $\tilde{p} \in S$ and $\tilde{p}'$ having the same weight (\ie, $w(\tilde{p}) = w(\tilde{p}')$).
    \item An arbitrary point $z$.
    \item The weighted medians $q$ ($q'$) of $S$ ($S'$) closest to $z$.
\end{compactitem}
If $q \neq q'$ then the path connecting $\tilde{p}$ and $\tilde{p}'$ goes through both $q$ and $q'$, and $q$ is closer to $\tilde{p}$ than $q'$.
\end{observation}
\begin{proof}
Assume $q \neq q'$. If $q'$ is also a weighted median of $S$ than $q$ must be closer to $z$ than $q'$. Similarly, if $q$ is also a weighted median of $S'$ than $q'$ must be closer to $z$ than $q$. As it is not possible that both $q$ is closer to $z$ than $q'$, and $q'$ is closer to $z$ than $q$, it must be the case that either $q$ is not a weighted median of $S'$, or $q'$ is not a weighted median of $S$. We assume without loss of generality that first option is true. Assume towards a contradiction that $q$ is not on the path from $\tilde{p}$ to $\tilde{p}'$. Let $P'$ be a sub multi-set of weighted points from $S'$ that can be connected via paths that do not contain $q$. If $f(\tilde{p}', P') < f(\tilde{p}', S')$ then $w(P') \leq w(S) / 2 = w(S') / 2$ because $q$ is a weighted median of $S$. Otherwise, consider the multi-set $P = P' \setminus \{\tilde{p}'\} \cup \{\tilde{p}\}$. By our assumption that the path between $\tilde{p}$ and $\tilde{p}'$ does not contain $q$, we get that $P$ is a sub multi-set of weighted points from $S$ connect by paths that do not include $q$. Hence, $w(P') = w(P) \leq w(S) / 2 = w(S') / 2$. In conclusion, we got that the weight of $P'$ is always at most $w(S') / 2$, and therefore, $q$ is a weighted median of $S'$, which contradicts our assumption. Hence, we can assume from now on that $q$ is on the path between $\tilde{p}$ and $\tilde{p}'$.

Since $q$ is not a weighted median of $S'$, there exists a sub multi-set $S''$ of weighted points from $S'$ that are connected via paths that do not go through $q$ and obeys $w(S'') > w(S') / 2$. Since $q$ is a weighted median of $S$, we must have: $f(\tilde{p}', S'') = f(\tilde{p}', S')$. Assume now, towards a contradiction, that $q'$ is on the path between $\tilde{p}$ and $q$. This assumption implies that $q$ is between $q'$ and $\tilde{p}'$, and therefore, the points of $S''$ are connected by paths that do not go through $q'$. However, since $q'$ is a weighted median of $S''$, this implies $w(S'') \leq w(S') / 2$, which is a contraction. Hence, we can assume from now on that $q'$ is not on the path between $\tilde{p}$ and $q$.

Next, assume towards a contradiction, that $q'$ is not on the path between $q$ and $\tilde{p}'$. Let $t$ be the point on the path between $q$ and $\tilde{p}'$ closest to $q'$. Consider an arbitrary sub multi-set $P$ of weighted points from $S'$ connected via paths that do not go through $t$. If the points of $P$ can be connected via paths that do not go through $q'$, then $w(P) \leq w(S') / 2$ because $q'$ is a weighted median of $S'$. Otherwise, at least one weighted point of $P$ is separated from $t$ by $q'$, and therefore, $\tilde{p}' \not \in P$. Moreover, all the paths between weighted points of $P$ do not go through $q$, and since $q$ is a weighted median of $S$, we get: $w(P) \leq w(S) / 2 = w(S') / 2$. In conclusion, every sub multi-set $P$ of weighted points from $S'$ connected via paths that do not go through $t$ has a weight of at most $w(S') / 2$, and therefore, $t$ is also a weighted median of $S'$.
By definition $q'$ is the weighted median of $S'$ closest to $z$, and therefore, $q'$ is closer to $z$ than $t$.

From the last conclusion we learn that the path from $q$ to $z$ goes through $t$ and $q'$. Hence, $q'$ is not a weighted median of $S$. Thus, there exists a sub multi-set $P$ of weighted points from $S$ that are connected via paths that do not go through $q'$, and obeys $w(P) > w(S) / 2$. Since $q'$ is a weighted median of $S'$, we must have $f(\tilde{p}, P) = f(\tilde{p}, S)$. Consider the set $P' = P \setminus \{\tilde{p}\} \cup \{\tilde{p}'\}$. Since $\tilde{p}$ and $\tilde{p}'$ are connect by a path that does not go through $q'$, the points of $P'$ are also connected by paths that do not go through $q'$. On the other hand, $q'$ is a weighted median of $S'$, and $P' \subseteq S'$, and therefore, $w(P) = w(P') \leq w(S') / 2 = w(S) / 2$, which is a contradiction. Hence, $q'$ must be on the path connecting $q$ and $\tilde{p}'$, which completes the proof of the observation.
\end{proof}

Using the above observation we can now show that WMM is two-sided IC. The agent-side IC of WMM is an easy new result, while the mediator-side IC is known from \cite{DekelFP2010,ProcacciaT13} for lines, and we extend it to trees. 

\begin{lemma}
The Weighted Median Mechanism is agent-side IC.
\end{lemma}
\begin{proof}
Consider an arbitrary agent $a_{i, j}$. We would like to show that given that $d_i$ is truthful, it is non-beneficial for $a_{i, j}$ to report a fake location $t'_{i, j}$ instead of its true location $t_{i, j}$. Let $\ell_i$ and $\ell'_i$ denote the median of the reports of agents $\{a_{i, 1}, a_{i, 2}, \ldots, a_{i,n _i}\}$ when $a_{i, j}$ reports its true location $t_i$ and the fake location $t'_i$, respectively. If $\ell_i = \ell'_i$, then reporting the fake location $t'_{i, j}$ does not affect the location picked by the center, and therefore, is non-beneficial for $a_{i, j}$. Otherwise, by Observation~\ref{ob:point_moves}, $\ell_i$ is on the path between $t_{i, j}$ and $\ell'_i$.

Let $F$ and $F'$ be the locations that the center picks for the facility assuming $a_{i, j}$ reports a location of $t_{i, j}$ or $t'_{i, j}$ to her mediator, respectively. If $F = F'$, then reporting the fake location $t'_{i, j}$ is, again, non-beneficial for $a_{i, j}$. Otherwise, by Observation~\ref{ob:point_moves}, $F$ and $F'$ are both on the path between $\ell_i$ and $\ell'_i$, and $F$ is closer to $\ell_i$ than $F'$.
Combining both results, we get that the path from $t_{i, j}$ to $\ell'_i$ goes first through $\ell_i$, then through $F$ and finally through $F'$. Thus, $\dist(t_{i, j}, F) < \dist(t_{i, j}, F')$, and it is strictly non-beneficial for $a_{i, j}$ to report the fake location $t'_{i, j}$.
\end{proof}

\begin{lemma}
The Weighted Median Mechanism is mediator-side IC.
\end{lemma}
\begin{proof}
Consider an arbitrary mediator $d_i$. We should prove that it is always a best response for $d_i$ to be truthful given that the agents of $A_i$ are truthful. Let $\ell_i$ be the median of the locations reported by the agents to $d_i$, that is closest to $z_i$.
Since the agents of $A_i$ are truthful, $\ell_i$ is also the median of the true locations of the agents of $A_i$ that is closest to $z_i$. Every strategy of $d_i$ can be characterized by the median closest to $z_i$ of the locations reported by $d_i$ on behalf of his agents because this is the only information used by the center. Hence, we can characterize any deviation of $d_i$ from truthfulness by the median $\ell'_i$ of the locations it reports that is closest to $z_i$.
Our objective is to prove that a deviation represented by $\ell'_i$ is no better than being truthful.

First, if $\ell_i = \ell'_i$ then the location picked by the center is unaffected by the deviation, and therefore, the deviation is non-beneficial. Thus, we can assume from now on $\ell_i \neq \ell'_i$. Next, let $F$ and $F'$ be the locations picked by the center given that $d_i$ reports locations whose medians (picked by the center according to the tie breaking rule) are $\ell_i$ and $\ell'_i$, respectively. If $F = F'$ then the deviation is clearly non-beneficial. Otherwise, by Observation~\ref{ob:point_moves}, the path from $\ell_i$ to $\ell'_i$ goes through $F$ and $F'$ in that order.

For every $x \in [0, 1]$, let $p_x$ be the point along the path between $\ell_i$ and $\ell'_i$ obeying $\dist(\ell_i, p_x) = x \cdot \dist(\ell_i, \ell'_i)$, and let $c_x$ be the social cost of the agents of $A_i$ if the facility is located at $p_x$. Let $p$ be an arbitrary point, and let us observe the change in $\dist(p_x, p)$ as $x$ increases. Clearly, $\dist(p_x, p)$ initially decreases at a constant rate and then starts to increase again in a constant rate (possibly the decrease part or increase part is of zero length).
Hence, the derivative of $c_x$ is a non-decreasing function in $x$ (except for a finite number of $x$ values for which it might not be defined). For $x = 0$, $p_x$ is the median of the locations of the agents of $d_i$, and therefore, $c_x$ is minimized for $x = 0$. Thus, the derivative of $c_x$ is non-negative for every $x > 0$, whenever it is defined. Let $y$ and $y'$ be the values of $x$ for which $F = p_y$ and $F' = p_{y'}$. Then, we have $y < y'$, and therefore, $c_y \leq c_{y'}$, \ie, locating the facility at $F'$ induces no better social cost for the agents of $A_i$ in comparison to locating it at $F$. Hence, reporting a set of locations whose median is $\ell'_i$ instead of a truthful report (which makes $\ell_i$ the median closest to $z_i$) is non-beneficial from $d_i$'s perspective.
\end{proof}

To complete the proof of Theorem~\ref{thm:wmm_properties}, we still need to show that WMM is $3$-competitive. This was already observed by~\cite{DekelFP2010,ProcacciaT13} for lines, but not for trees. In order to prove this result for trees, we observe that WMM is a restriction of a more general mechanism called IWMM which we analyze in Section~\ref{sec:hierarchy}. The competitive ratio of WMM follows immediately from the result proved by Proposition~\ref{prop:IWMM-comp} for IWMM.
\section{Randomized Two-Sided IC Mechanisms} \label{sc:rand}
It is known by \cite{DekelFP2010,ProcacciaT13} that there is no deterministic Two-Sided IC Mechanism that is better than $3$-competitive (this result also appears as Theorem~\ref{th:lower-bound-det-mech}). In this section we show that we can improve and achieve a competitive ratio of $2$ by switching to randomized mechanisms, and that this is the best ratio that can be achieved.

To simplify the exposition of our mechanism, we first describe it for the simple case of line metrics (\ie, for the case where the tree $T$ is simply an interval). A line metric is the Euclidean metric of an arbitrary interval $[a, b]$ (where $a < b$ are real numbers). Notice that a point in the metric is simply a real number from the interval $[a, b]$.
The {\em Two Percentiles Range Mechanism (TPRM)} is a direct revelation mechanism in which the center runs the following algorithm:
\begin{compactitem}
    \item For each mediator $d_i$ compute the median $\ell_i$ of the multi-set $\{t'_{i, j} | 1 \leq j \leq n_i\}$ that is closest to $z_i$, where $t'_{i, j}$ is the location reported by $d_i$ for agent $a_{i, j}$ and $z_i$ is an arbitrary point chosen independently of the reports received from the agents.
    \item Consider the multi-set $S$ of points, created by adding $\ell_i$ to the multi-set $n_i$ times, for each $i$. Let $u_i$ denote the $i$-th element of this multi-set when sorted in any non-decreasing order.
    \item Randomly choose a location for the facility from the list: $u_{\lfloor n/4 \rfloor + 1}, u_{\lfloor n/4 \rfloor + 2}, \ldots, u_{\lceil 3n/4 \rceil}$, where the probability of each value $u_i$ in this list is $(n/2)^{-1}$, except for the first and last values ($u_{\lfloor n/4 \rfloor + 1}$ and $u_{\lceil 3n/4 \rceil}$), which have a probability of $(1 - r / 4) / (n/2)$ where $r$ is the reminder of dividing $n$ by $4$.
\end{compactitem}
Like in the deterministic case, Observation~\ref{ob:closest_median_exists} ensures that $\ell_i$ is well defined for every $i$.
Also similarly to the deterministic case, this direct revelation mechanism can also be executed with much less communication by only asking each mediator to report a single point (the location of the median closest to some arbitrary point) and running a mediator based algorithm that corresponds to the two final steps of TPRM on the reports. The resulting mechanism will have the same properties (competitiveness, incentives) as the direct revelation mechanism.

\begin{theorem} \label{th:tprm_two_sided_and_competitive}
For any line metric, the \emph{Two Percentiles Range Mechanism} (TPRM) is a randomized two-sided IC mechanism with a competitive ratio of $2$.
\end{theorem}

The proof of Theorem~\ref{th:tprm_two_sided_and_competitive} is deferred to Section~\ref{ssc:TPRM-analysis}. The intuition behind the improved competitive ratio of TPRM, as compared to WMM, is as follows. Consider a section $s$ connecting the real locations of two agents separated by no other agent.
\begin{compactitem}
    \item If $s$ is located near the median of all the agents, then it is not very important on which side of it is the facility located, because either way $s$ will contribute to the distance between the facility and about half the agents.
    \item If $s$ is located very far from the median, then both TPRM and WMM will locate the facility at the ``right'' side of $s$.
    \item If $s$ is not located near the median, nor very far from it, then it is important that the mechanism will locate the facility at the right side of $s$. WMM makes the wrong call for some inputs, and, being deterministic, when it makes the wrong call it makes it with probability $1$. On the other hand, it can be shown that the randomized TPRM always has a significant probability of making the right choice.
\end{compactitem}

Next, we present a mechanism that extends TPRM to general trees, called the \emph{Tree Randomized Mechanism (TRM)}. Like TPRM, TRM defines a distribution over the medians of the sets reported by the mediators, and then chooses the facility location randomly according to this distribution. The distribution is carefully picked to achieve the best possible competitive ratio of 2. More specifically, the mass of the distribution is placed on the ``central parts'' of the tree, which is analogous to behavior of TPRM. Formally, TRM is a direct revelation two-sided IC mechanism that is also $2$-competitive. The description of TRM consists of two parts. The first part of TRM determines the medians of the sets reported by the mediators, and then splits the edges of $T$ at these points, which allows us to assume that all medians are vertexes of $T$. Additionally, this part finds a point $r$ which is considered to be the ``center'' of $T$. If necessary, an edge is split to make $r$ a vertex of $T$, and then $T$ is rooted at $r$.

\begin{tabbing}
\textbf{TRM - Part $1$}\\
1. \= Fix an arbitrary point $z$ independently of the reports of the mediators.\\
2. \> Define $\size(p) = \sum_{i|\ell_i = p} n_i$ for every point $p$, where $\ell_i$ is the median closest to $z$ of the points\\ \> reported by $d_i$.\\
3. \> Let $L = \{(p, \size(p)) ~|~ \mbox{$p$ is a point and $\size(p) > 0$}\}$.\\
4. \> Let $r$ be the weighted median of $L$ closest to z.\\
5. \> For \= every point $p \in \{r\} \cup \{\ell_i ~|~ 1 \leq i \leq k\}$ do:\\
6. \> \> If \= $p$ is an internal point of an edge $e$, then:\\
7. \> \> \> Split $e$ at point $p$ to two edges.\\
8. \> Root $T$ at $r$.
\end{tabbing}

Clearly we can assume after the first part of TRM that $\{r\} \cup \{\ell_i ~|~ 1 \leq i \leq k\} \subseteq V$. Since $T$ is now rooted, we can define additional notation that is used to describe the second part of TRM.
\begin{definition}
~
\begin{compactitem}
    \item $\children(u)$ - The set of children nodes of $u$ in the tree.
    \item $\subtree(u)$ - The set of nodes in the subtree of node $u$, including $u$ itself. More formally, $\subtree(u) = \{u\} \cup [\cup_{u' \in \children(u)} \subtree(u')]$.
    \item $\treesize(u)$ - The number of agents represented by a mediator $d_i$ for which $\ell_i \in \subtree(u)$. More formally, $\treesize(u) = \size(u) + \sum_{u' \in \children(u)} \treesize(u')$.
\end{compactitem}
\end{definition}

The second part of TRM defines the probability distribution used to select the facility location. Informally, this distribution has the following property: the probability that the facility location is in the \emph{subtree} of a (non-root) vertex $u$ is proportional to $\max\{\treesize(u) - n/4, 0\}$. An alternative view of this distribution is that the algorithm preprocess the agents by discarding some agents of peripheral vertexes, and ``promoting'' other agents of these vertexes to be associated with vertexes which are closer to the root. After this preprocessing, the algorithm selects a uniformly random remaining agent, and locates the facility at the location desired by the mediator associated with the selected agent. The preprocessing is done in a way that guarantees the following:
\begin{compactitem}
	\item If only a few agents originally belong to the subtree of a vertex $u$ (\ie, $\treesize(u) \leq n/4$), then all these agents are either discarded or promoted to some ancestor vertex of $u$, and thus, $u$ never becomes the facility location.
	\item If many agents originally belong to the subtree of a vertex $u$, then the preprocessing reduces the number of agents associated with this subtree by $n/4$ (or $n/2$ if $u$ is the root).
\end{compactitem}

\begin{tabbing}
\textbf{TRM - Part $2$}\\
1. \= Let $X = \{u \in V ~|~ \treesize(u) \geq n/4\}$.\\
2. \> Let $c(u)$ denote the value $n/4$ for every node $u \in X - \{r\}$ and $n/2$ for $u = r$.\\
3. \> Pick \= every node $u \in X$ as the facility location with7 probability:
\end{tabbing}
\vspace{-0.2in}
\[
    p(u)
    =
    \frac{\treesize(u) - c(u)}{n/2} - \frac{\sum_{u' \in X \cap \children(u)}(\treesize(u') - c(u'))}{n/2}
    \enspace.
\]
For consistency, define $p(u) = 0$ for every point $u\notin X$.

Note that this direct revelation mechanism can also be viewed (and executed) as a mediator based algorithm as the center only needs from each mediator $d_i$ a location $\ell_i$ which is optimal from $d_i$'s perspective. The proof of the next theorem is deferred to Appendix~\ref{app:rand_proofs}.

\begin{theorem} \label{th:general_metrics_mechanism}
For any tree metric, the \emph{Tree Randomized Mechanism} (TRM) is a randomized two-sided IC mechanism with a competitive ratio of $2$.
\end{theorem}

\begin{figure}[t]
    \center
    \framebox{\includegraphics[width=2.5in, height=1.5in, clip=true, trim=0.8in 1in 0in 0in]{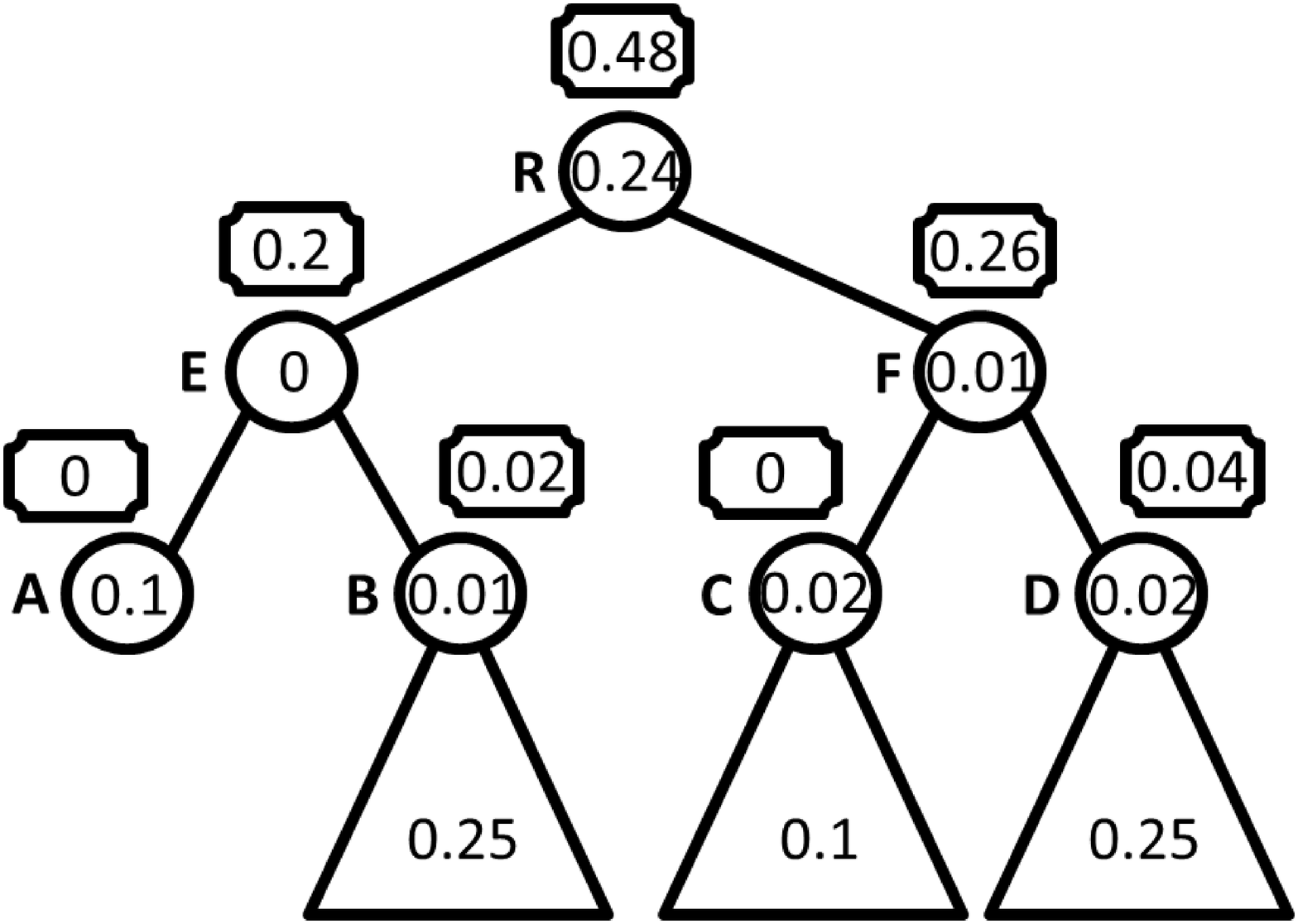}}
    \caption{Example of a probability distribution induced by TRM. \label{fig:tree_example}}
\end{figure}

Let us get a better understanding of TRM by considering an example input (given as Figure~\ref{fig:tree_example}) and explaining the probability distribution induced by TRM. The figure depicts the top $7$ nodes of an example tree $T$ that can be outputted by the first part of TRM. The number inside each node represents the portion of the agents population represented by mediators whose (sole) median is this node. For example, the number $0.24$ appears inside the root node $R$, hence, in our example, $0.24 \cdot n$ agents are represented by mediators whose median is $R$. Some of the nodes in the figure have a triangle shape dangling from them. Each triangle represents the subtree of the node it is dangling from, and the number written inside it represents the portion of the agents represented by mediators whose (sole) median is inside the subtree (but is not the root of the subtree). For example, inside the rightmost triangle we have the label $0.25$. Hence there are $0.25 \cdot n$ agents represented by mediators whose median is inside the subtree of $D$, but is not $D$ itself. Finally, outside of each node there is an additional number (in a box). This number represents the probability that TRM selects this node as the facility location.

Following is a short explanation of how (some of) the probabilities in Figure~\ref{fig:tree_example} were calculated. We say that a node (subtree) in Figure~\ref{fig:tree_example} has a weight of $x$ if $x$ agents are represented by mediators whose median is the node (a node in the subtree). Observe that TRM selects $R$ to be the root node $r$ since its left subtree has a weight of only $0.36 \cdot n < n/2$, and its right subtree has a weight of only $0.4 \cdot n < n/2$. Next, observe that no triangle in the figure represents a subtree with weight of more than $0.25$, and therefore, no node in the subtrees represented by these rectangles has a positive probability to be the location of the facility. Consider now a few of the nodes of Figure~\ref{fig:tree_example}.
\begin{compactitem}
    \item The subtree rooted at $A$ (which contains $A$ alone) has a weight of only $0.1 \cdot n \leq 0.25 \cdot n$, and therefore, $A\notin X$ and thus has $0$ probability to be the facility location.
    \item The subtree rooted at $B$ has a weight of $0.26 \cdot n > 0.25 \cdot n$, and no other node in this subtree has a positive probability to be the facility location. Therefore, $B$ has a probability of $2(0.26 - c(B))=2(0.26 - 0.25)=0.02$ to be the facility location.
    \item The subtree rooted at $E$ has a weight of $0.36 \cdot n > 0.25 \cdot n$. However, $E$ has a single child $B$ associated with a positive probability, and the subtree rooted at $B$ has a weight of $0.26$. Therefore, $E$ has a probability of $2(0.36 - c(E) - (0.26 - c(B)))=0.2$ to be the facility location.
    \item The subtree rooted at $R$ has a weight of $1 \cdot n$. However, $R$ has two children $E$ and $F$, each associated with a positive probability by the algorithm, whose subtrees have weights of $0.36$ and $0.4$, respectively. Therefore, $R$ has a probability of $2(1 - c(R) - (0.36 - c(E)) - (0.4 - c(F)))=0.48$ to be the facility location.
\end{compactitem}

The following theorem shows that TRM and TPRM have optimal competitive ratios, its proof is deferred to Section~\ref{sec:lower}.

\begin{theorem}
\label{th:lower-bound-rand-mech}
Fix any constant $\ee > 0$.
Then, no direct revelation randomized mechanism that is mediator-side incentive compatible and no mediator based algorithm has a competitive ratio of $2-\ee$, even for line metrics.
\end{theorem}

\subsection{Analysis of TPRM} \label{ssc:TPRM-analysis}

In this subsection we analyze TPRM and prove Theorem~\ref{th:tprm_two_sided_and_competitive}. We begin the proof by establishing some notation.
Let us sort the agents according to their location (breaking ties arbitrarily), let $b_i$ be the $i^{th}$ agent in this order and let $t_i$ be the location of $b_i$. We say that {\em agent $b_i$ is to the left (right) of agent $b_j$} if $i < j$ ($i > j$). Let $s_i$ denote the segment connecting $b_i$ and $b_{i + 1}$, and let $x_i = t_{i + 1} - t_i$ be the length of $s_i$. We say that {\em segment $s_i$ is to the left (right) of agent $b$} if $b_i$ ($b_{i + 1}$) is to the left (right) of $b$. Given a location for the facility, we say that {\em an agent $b_i$ uses segment $s_i$} if $s_i$ is on the path between the agent and the given location.

Before analyzing TPRM, let us make sure it is well defined.

\begin{observation}
Let $p_i$ be the probability that TPRM picks $u_i$ as the location of the facility. Then, $\sum_{i = 1}^n p_i = 1$.
\end{observation}
\begin{proof}
The number of elements in the list $u_{\lfloor n/4 \rfloor + 1}, u_{\lfloor n/4 \rfloor + 3}, \ldots, u_{\lceil 3n/4 \rceil}$ is:
\[
    \lceil 3n/4 \rceil - \lfloor n/4 \rfloor
    =
    (3n/4 + r/4) - (n/4 - r/4)
    =
    n/2 + r/2 \enspace.
\]

Hence, $n/2 + r/2 - 2$ elements of the list have a probability of $(n/2)^{-1}$, and two have a probability of $(1 - r / 4) / (n/2)$. Summing up all these probabilities, we get:
\[
    \sum_{i = 1}^n p_i
    =
    \frac{n/2 + r/2 - 2}{n/2} + 2 \cdot \frac{1 - r / 4}{n/2}
    =
    1 \enspace.
    \qedhere
\]
\end{proof}

\subsubsection{Proof: TPRM is \texorpdfstring{$2$}{2}-Competitive}

First we analyze the competitive ratio of TPRM as an algorithm (when players are truthful).
The following lemma characterizes the social cost of the optimal facility location. This characterization is used later to prove the competitive ratio of TPRM.

\begin{lemma}
The social cost of the optimal facility location is:
\begin{equation} \label{eq:opt_value}
    \sum_{i = 1}^{\lfloor (n - 1)/ 2 \rfloor} i \cdot x_i + \sum_{i = \lfloor (n + 1)/2 \rfloor}^{n - 1} (n - i) \cdot x_i \enspace.
\end{equation}
\end{lemma}
\begin{proof}
By Observation~\ref{ob:median_optimal}, $t_{\lceil n / 2 \rceil}$ is an optimal facility location. Let us evaluate the social cost of placing the facility at this location. Consider segment $s_i$. If $1\leq i < n/2$, then $s_i$ is to the left of $b_{\lceil n / 2 \rceil}$, and therefore, $i$ agents use $s_i$: $b_1, b_2, \ldots, b_i$. Thus, the contribution of $s_i$ to the cost of the optimal solution is $i \cdot x_i$. On the other hand, if $n-1\geq i \geq n/2$ (\ie, $s_i$ is to the right of $b_{\lceil n / 2 \rceil}$), then $n - i$ agents use $s_i$: $b_{i + 1}, b_{i + 2}, \ldots, b_n$. Thus, the contribution of $s_i$ to the cost of the optimal solution is $(n - i) \cdot x_i$. Adding up all the above contributions, we get Expression~(\ref{eq:opt_value}).
\end{proof}

Pick an arbitrary segment $s_i=[t_i,t_{i+1}]$ (recall that $t_i$ is the location of agent $b_i$), and let $L_i$ and $R_i$ denote the number of $u_j$ values that are smaller than $t_{i + 1}$ and larger then $t_i$, respectively. Intuitively, $L_i$ and $R_i$ denote the number of agents represented by a mediator $d_h$ for which $\ell_h$ is to the left and right of $s_i$, respectively.

\begin{lemma} \label{le:l_i_r_i_bound}
For every $i$, $L_i \leq 2i$ and $R_i \leq 2(n - i)$.
\end{lemma}
\begin{proof}
We prove $L_i \leq 2i$; the proof that $R_i \leq 2(n - i)$ is symmetrical. Let $D_i = \{d_h ~|~ \ell_h < t_{i + 1}\}$. Notice that a mediator $d_h$ contributes to $L_i$ if and only if $\ell_h < t_{i + 1}$ (\ie, $\ell_h$ is to the left of $s_i$), in which case his contribution to $L_i$ is $n_h$. Hence, $L_i = \sum_{d_h \in D_i} n_h$.

Consider now an arbitrary mediator $d_h \in D_i$. At least half the agents of $d_h$ must have locations smaller or equal to $\ell_h < t_{i + 1}$, and therefore must be of the form $b_j$ for $j \leq i$. Hence, the size of the set $B_i = \{b_j ~|~ j \leq i \wedge \exists_{h \in D_i} b_j \in A_h\}$ is at least $\sum_{d_h \in D} n_h/2 = L_i / 2$. The lemma now follows by observing that $B$ cannot contain more than $i$ agents.
\end{proof}

We are now ready to prove the approximation ratio of TPRM.

\begin{lemma}
The Two Percentiles Range Mechanism is $2$-competitive.
\end{lemma}
\begin{proof}
Let $c_i$ be the coefficient of $x_i$ in Expression~(\ref{eq:opt_value}), \ie, $c_i = i$ for $i \leq n/2$, and $c_i = n - i$ for $i \geq n / 2$. To prove that TPRM is a $2$-competitive mechanism, it is enough to show that, given the facility location selected by TPRM, the expected number of agents using segment $s_i$ is at most $2c_i$.

Consider an arbitrary section $s_i$, and assume $i \leq n/2$ (the case $i \geq n/2$ is symmetric). Let $P_i$ be the probability that the facility location chosen by TPRM is to the left of $s_i$. If the chosen facility location is indeed to the left of $s_i$, then all the agents to the right of $s_i$ will use it, \ie, $n - i$ agents. On the other hand, if $s_i$ is to the right of the chosen facility location, than only $i$ agents will use it. Hence, the expected number of agents using $s_i$ is:
\[
    P_i \cdot (n - i) + (1 - P_i) \cdot i
    =
    P_i \cdot (n - 2i) + i
    \enspace.
\]
The ratio between this expectation and $c_i = i$ is:
\begin{equation} \label{eq:ratio}
    P_i \cdot \left(\frac{n}{i} - 2\right) + 1 \enspace.
\end{equation}
Proving the lemma now boils down to proving that (\ref{eq:ratio}) is always upper bounded by $2$. Notice that the coefficient of $P_i$ is always non-negative (recall that we assumed $i \leq n/2$), and therefore, any upper bound on $P_i$ will translate into an upper bound on (\ref{eq:ratio}).

By Lemma~\ref{le:l_i_r_i_bound}, $L_i \leq 2i$, which implies, by definition, that $u_{2i + 1} \geq t_{i + 1}$ (\ie, $u_{2i + 1}$ is to the right of $s_i$). Hence, the facility location picked can be to the left of $s_i$ only if the center picks one of the values: $u_{\lfloor n/4 \rfloor + 1}, u_{\lfloor n/4 \rfloor + 2}, \ldots, u_{2i}$. Thus, we can upper bound $P_i$ by:
\[
    P_i
    \leq
    \max \left\{0, \frac{[2i - (\lfloor n/4 \rfloor + 1)] + (1 - r/4)}{n/2} \right\}
    =
    \max \left\{0, \frac{2i - n/4}{n/2} \right\}
    =
    \max \left\{0, \frac{4i}{n} - \frac{1}{2} \right\}
    \enspace.
\]

If $\frac{4i}{n} - \frac{1}{2}  \leq 0$, than $P_i  = 0$, and the ratio (\ref{eq:ratio}) becomes $1$.
Hence, we only need to consider the case that
$\frac{4i}{n} - \frac{1}{2} > 0$. Using this assumption, the ratio~(\ref{eq:ratio}) can be upper bounded by:
\[
     P_i \cdot \left(\frac{n}{i} - 2\right) + 1
     \leq
     \left( \frac{4i}{n} - \frac{1}{2} \right) \cdot \left(\frac{n}{i} - 2\right) + 1
     =
     4 - \frac{8i}{n} - \frac{n}{2i} + 1 + 1
     =
     6 - \frac{8i}{n} - \frac{n}{2i}
     \enspace.
\]

It can be easily checked that the last expression is maximized for $i = n/4$, and its maximum value is exactly $2$.
\end{proof}

\subsubsection{Proof: TPRM is Two-Sided IC}

Next, we prove that TPRM is two-sided incentive compatible. For that purpose we need the following observation.
\begin{observation} \label{ob:move_mediator_bad_for_agent}
Fix the reports of all mediators except for $d_i$. Let us denote by $u_1, u_2, \ldots, u_n$ and $u'_1, u'_2, \ldots, u'_n$ the set $\{u_h\}_{h = 1}^n$ assuming mediator $d_i$ reports a set of locations whose median is $\ell_i$ and $\ell'_i$,
respectively, where $\ell'_i > \ell_i$ ($\ell'_i < \ell_i$). Then $u'_h \geq u_h$ ($u'_h \leq u_h$) for every $1 \leq h \leq n$. Moreover, if $u_h < \ell_i$ ($u_h > \ell_i$) then $u_h = u'_h$.
\end{observation}
\begin{proof}
We prove the case of $\ell'_i > \ell_i$. The other case is symmetric. Assume the median of the reports of $d_i$ is $\ell_i$, and let us analyze the effect on the list $u_1, u_2, \ldots, u_n$ when this median changes to $\ell'_i$. Following the change some $u_h$ values that are equal to $\ell_i$ increase to $\ell'_i$, and possibly move to later places in the list.
From this description, it is clear that $u_h \leq u'_h$ for every $1 \leq h \leq n$. Moreover, the part of the list containing values smaller than $\ell_i$ is unaffect, and therefore, $u_h = u'_h$ whenever $u_h < \ell_i$.
\end{proof}

\begin{lemma}
The Two Percentiles Range Mechanism is agent-side IC.
\end{lemma}
\begin{proof}
Consider an arbitrary agent $a_{i, j}$. We would like to show that, given that $d_i$ is truthful, it is non-beneficial for $a_{i, j}$ to report a fake location $t'_{i, j}$ instead of its true location $t_{i, j}$. Let $\ell_i$ and $\ell'_i$ denote the median of the reports (breaking ties towards $z_i$) of mediator $d_i$ assuming $a_{i, j}$ reports its true location $t_i$ and the fake location $t'_i$, respectively. If $\ell_i = \ell'_i$, then reporting the fake location $t'_{i, j}$ is certainly non-beneficial for $a_{i, j}$. Otherwise, assume without loss of generality that $\ell_i < \ell'_i$. By Observation~\ref{ob:point_moves}, $\ell_i$ is on the path between $t_{i, j}$ and $\ell'_i$, \ie, $t_{i, j} \leq \ell_i$.

Let $u_1, u_2, \ldots, u_n$ and $u'_1, u'_2, \ldots, u'_n$ denote the $u_h$'s before and after the deviation of $a_{i, j}$. By Observation~\ref{ob:move_mediator_bad_for_agent}, $u'_h \geq u_h$ for every $1 \leq h \leq n$ and $u_h = u'_h$ when $u_h < \ell_i$. Hence, $\dist(u_h, t_{i, j}) \geq \dist(u'_h, t_{i, j})$ for every $1 \leq h \leq n$ because $t_{i, j} \leq \ell_i$. As TPRM selects $u_h$ and $u'_h$ as the facility location with the same probability, we get that the expected distance of $t_{i, j}$ from the facility cannot decrease following the deviation of $a_{i,j}$.
\end{proof}

\begin{lemma}
The Two Percentiles Range Mechanism is mediator-side IC.
\end{lemma}
\begin{proof}
Consider an arbitrary mediator $d_i$. We should prove that it is always a best response for $d_i$ to be truthful given that the agents of $A_i$ are truthful. Let $\ell_i$ be the median of the locations reported by the agents to $d_i$, that is closest to $z_i$.
Since the agents of $A_i$ are truthful, $\ell_i$ is also the median of the true locations of the agents of $A_i$ that is closest to $z_i$. Every strategy of $d_i$ can be characterized by the median closest to $z_i$ of the locations reported by $d_i$ on behalf of his agents because this is the only information used by the center. Hence, we can characterize any deviation of $d_i$ from truthfulness by the median $\ell'_i$ of the locations it reports that is closest to $z_i$.
Our objective is to prove that a deviation represented by $\ell'_i$ is no better than being truthful.

Assume, without loss of generality, $\ell_i < \ell'_i$. Let $u_h$ and $u'_h$ be the value of the $h^{th}$ element in the list $u_1, u_2, \ldots, u_n$, assuming $d_i$ is truthful or chooses a deviation represented by $\ell'_i$. If $u_h = u'_h$, then $d_i$ clearly does not care whether the facility is located at $u_h$ or $u'_h$. Otherwise, by Observation~\ref{ob:move_mediator_bad_for_agent}, $u'_h > u_h \geq \ell_i$.

Let $p_x = \ell_i + x$, and let $c_x$ be the social cost of the agents represented by $d_i$ if the facility is located at $p_x$. Let $p$ be an arbitrary point, and let us observe the change in $\dist(p_x, p)$ as $x$ increases. Clearly, $\dist(p_x, p)$ initially decreases at a constant rate and then starts to increase again in a constant rate (possibly the decrease or increase part is of zero length). Hence, the derivative of $c_x$ is a non-decreasing function in $x$ (expect for a finite number of $x$ values for which it might not be defined). For $x = 0$, $p_x$ is a median of the locations of the agents of $d_i$, and therefore, $c_x$ is minimized for $x = 0$. Thus, the derivative of $c_x$ is non-negative for every $x > 0$, whenever it is defined. Let $y$ and $y'$ be the values of $x$ for which $u_h = p_y$ and $u'_h = p_{y'}$. Then, we have $y < y'$, and therefore, $c_y \leq c_{y'}$, \ie, locating the facility at $u'_h$ induces no better social cost for the agents of $A_i$ in comparison to locating it at $u_h$. Hence, mediator $d_i$ does not prefer that the facility will be located at $u'_h$ over locating it at $u_h$.

Consider the following process. We start with $u_1, u_2, \ldots, u_n$, and at each step pick some $u_h$ and replace it with $u'_h$ (\ie, the probability mass that was given to $u_h$ moves to $u'_h$). Notice two things. First, $u_h$ and $u'_h$ have the same probability to be selected by the center (since they have the same index), and therefore, no step improves the social cost of the agents represented by $d_i$. 
Second, the process described starts with the configuration resulting from a truthful $d_i$, and ends with the configuration resulting of a deviation represented by $\ell'_i$. Hence, this deviation is non-beneficial for $d_i$.
\end{proof}
\section{Lower Bounds on Two-Sided IC Mechanisms}
\label{sec:lower}

In this section we prove the lower bounds given in Sections~\ref{sc:deterministic_algorithm} and~\ref{sc:rand} on the competitive ratios that two-sided incentive compatible mechanisms can achieve on any tree metric.
All our lower bounds are based on a construction for the line metric on $[0,1]$ (\ie, on the simple tree $T$ that is just an interval) which is given as Example~\ref{exmple:lower-bound}. The lower bounds exploit the inability of mediator-side IC mechanisms to use information other than the location of the optimal point (median) for each mediator, due to the mediator incentive constraints. Example~\ref{exmple:lower-bound} is essentially identical to the construction used in~\cite{DekelFP2010,ProcacciaT13}. However, while these papers only provide a lower bound for the deterministic case, we show that the same construction can also be used to prove a lower bound for the randomized case.
Note that interval lower bounds automatically extend to any tree $T$, since the interval $[0,1]$ can be mapped to any path connecting two leaves of the tree.

\begin{example}
\label{exmple:lower-bound}
Fix $r$ to be an arbitrary positive integer.
For every pair of points $h$ and $l$ such that $0 \leq l < h \leq 1$ we define three instances $I_1^{l,h}$, $I_2^{l,h}$ and $I_3^{l,h}$.
Each instance has two mediators $d_1$ and $d_2$, with each mediator representing $2r+1$ agents.
In $I_1^{l,h}$, mediator $d_1$ represents $r+1$ agents located at $l$ and $r$ agents located at $h$, while $d_2$ represents $2r+1$ agents located at $h$. In $I_2^{l,h}$, mediator $d_1$ represents $2r+1$ agents located at $l$, while $d_2$ represents $r$ agents located at $l$, and $r+1$ at $h$.
Finally, in $I_3^{l,h}$ mediator $d_1$ represents $2r+1$ agents located at $l$, while $d_2$ represents $2r+1$ agents located at $h$.
\end{example}

To make use of Example~\ref{exmple:lower-bound}, we need to make the following observation.

\begin{observation}
In all instances $I_1^{l,h}$, $I_2^{l,h}$ and $I_3^{l,h}$ of Example~\ref{exmple:lower-bound}, each mediator represents agents with a unique median. Moreover, the unique median of the agents represented by $d_1$ is $l$, while the unique median of the agents represented by $d_2$ is $h$. The globally optimal location for $I_1^{l,h}$ is to locate the facility at $h$, and has a cost of $(h - l)(r+1)$. Similarly, the globally optimal location for $I_2^{l,h}$ is to locate the facility at $l$, and its cost is also $(h - l)(r+1)$.
\end{observation}

Using the above observation, we can prove the following lemma.

\begin{lemma}
\label{lemma:IC-median}
Consider the two instances $I_1=I_1^{0,1}$ and $I_2=I_2^{0,1}$ defined as in Example~\ref{exmple:lower-bound}.
For any mediator-side IC mechanism, the expected location\footnote{Notice that a location in a line metric is just a number, hence, one can calculate expectation over locations.} picked given either $I_1$ or $I_2$ must be the same.
\end{lemma}
\begin{proof}
Denote $I_3=I_3^{0,1}$.
We show that the expected location picked given either $I_1$ or $I_2$ must be the same since in both cases the expected location picked is the expected location picked given $I_3$. We assume throughout the proof that the agents are truthful.

Let $x_j$ be the expected location of the picked location given input $I_j$ (for $j\in \{1,2,3\}$). For any facility location distribution $F$ whose support is within $[0,1]$, the expected cost for mediator $d_1$ in instance $I_1$ is $r+\mathbb{E}[x]$, while the expected cost for mediator $d_1$ in instance $I_3$ is $(2r+1)\mathbb{E}[x]$, where $\mathbb{E}[x]$ is the expected location picked according to $F$. We observe that for both instances and any distribution, the cost of $d_1$ is linearly increasing in the expected location of the facility. This means that for both inputs, mediator $d_1$ prefers a distribution with lower expected location.
That is, if $0\leq x_1< x_3$ then given $I_3$ mediator $d_1$ prefers to pretend that the input is $I_1$, and
if $0\leq x_3< x_1$ then given $I_1$ mediator $d_1$ prefers to pretend that the input is $I_3$.
We conclude that $x_1=x_3$. Similar arguments show that $x_2=x_3$, and thus, $x_1=x_2$.
\end{proof}

We are now ready to prove our lower bounds.

\begin{reptheorem}{th:lower-bound-det-mech}
Fix any constant $\ee > 0$.
Then, no direct revelation \emph{deterministic} mechanism that is mediator-side incentive compatible has a competitive ratio of $3-\ee$, even for line metrics.
\end{reptheorem}
\begin{proof}
Consider an arbitrary direct revelation deterministic mediator-side incentive compatible mechanism, and assume all players are truthful.
To prove the theorem we fix $r \geq 2\ee^{-1}$, and consider two instances $I_1=I_1^{0,1}$ and $I_2=I_2^{0,1}$ as defined in Example~\ref{exmple:lower-bound}.
By Lemma~\ref{lemma:IC-median} (with a single point distributions determined by the deterministic algorithm), the location picked given $I_1$ must be the same as the location picked given $I_2$. Let us denote this location by $p$.

If $p=0$ then for $I_1$ the cost is $3r+1$ while the optimal cost is $r+1$. The resulting competitive ratio is
$\frac{3r+1}{r+1}=3-\frac{2}{r+1}> 3-\frac{2}{r}\geq 3-\ee$ (as $r\geq 2\ee^{-1}$).
Similar arguments show that if $p=1$ then the competitive ratio is again worse than $3-\ee$ for $I_2$. This completes the proof for the case $p\in \{0,1\}$.

Assume now $p\in (0,1)$, and consider the location $x'_2\in [0,1]$ picked by the mechanism given the input $I'_2=I_2^{0,p}$.
If $x'_2< p$ then given $I_2$ mediator $d_2$ prefers to pretend that the input is $I'_2$, and
if $x'_2> p$ then given input $I'_2$ mediator $d_2$ prefers to pretend that the input is $I_2$. We conclude that $x'_2= p$.

The social cost of locating the facility at $p$ given input $I'_2$ is $(3r+1)p$, while the optimal social cost for this instance is $(r+1)p$. The competitive ratio achieved is, therefore, no more than:
$\frac{3r+1}{r+1}=3-\frac{2}{r+1}> 3-\frac{2}{r}\geq 3-\ee$ (as $r\geq 2\ee^{-1}$).
\end{proof}

The lower bound for the randomized case is very similar to the deterministic case, but it uses Lemma~\ref{le:single_better_than_distribution} which allows us to extend the result to randomized mechanisms. 

\begin{reptheorem}{th:lower-bound-rand-mech}
Fix any constant $\ee > 0$.
Then, no direct revelation randomized mechanism that is mediator-side incentive compatible and no mediator based algorithm has a competitive ratio of $2-\ee$, even for line metrics.
\end{reptheorem}
\begin{proof}
Consider an arbitrary direct revelation randomized mechanism which is mediator-side incentive compatible, and assume all players are truthful.
To prove the theorem we fix $r\geq\ee^{-1}$, and consider two instances $I_1=I_1^{0,1}$ and $I_2=I_2^{0,1}$ as defined in Example~\ref{exmple:lower-bound}. By Lemma~\ref{lemma:IC-median}, the expected location for $I_1$ must be the same as the expected location for $I_2$. Let $p$ be this expected location.

By Lemma~\ref{le:single_better_than_distribution}, for any set of agents, replacing a distribution whose expected location is $x_j$ by deterministically locating the facility at $x_j$ only decreases the cost. This means that the performance of the mechanism is no better than the worse performance (over $I_1$ and $I_2$) achieved by locating the facility deterministically at some $p\in [0,1]$. The cost of $I_1$ when the facility is located at $p$ is: $p(r + 1) + (1 - p)(3r + 1)$, and the cost of $I_2$ is $p(3r + 1) + (1 - p)(r + 1)$. The maximum of these costs is:
\begin{align*}
    \max \left\{p(r + 1) + (1 - p)(3r + 1), p(3r + 1) + (1 - p)(r + 1)\right\} \hspace{-3in} & \hspace{3in}\\
    \geq{} &
    \frac{[p(r + 1) + (1 - p)(3r + 1)] + [p(3r + 1) + (1 - p)(r + 1)]}{2}
    =
    2r + 1
    \enspace,
\end{align*}
which induces a competitive ratio of at least:
$\frac{2r+1}{r+1}=2-\frac{1}{r+1}> 2-\frac{1}{r}\geq 2-\ee$ (as $r\geq\ee^{-1}$).

This completes the proof for direct revelation randomized mechanisms that are mediator-side incentive compatible. The theorem holds also for mediator based algorithms because such algorithms cannot distinguish between the inputs $I_1^{0,1}$ and $I_2^{0,1}$.
\end{proof}
\section{Multiple Levels of Mediation}
\label{sec:hierarchy}

In this section we present results for the case of multiple levels of mediation. 
We can represent a hierarchy of mediators by a tree (note that this tree is {\em not} the same tree as the one defining the metric). The root of the tree represents the center, each internal node represents a mediator, and each leaf represents an agent.  The tree is common knowledge. Let $s$ be the {\em depth} of the tree, or the maximal number of edges between the root and a leaf.
The case we have studied so far is thus represented by a tree with three levels ($s$ = 2), the root represents the center, the internal level represents the mediators, and the leaves represents the agents.

As before, a {\em player} is either an agent or a mediator. For each player, another player is a {\em direct kin} if that other player is either a descendent in the tree, or an ancestor in the tree.
Recall that a direct revelation mechanism is a mechanism in which each agent reports her location, and each mediator reports all of the locations of the agents below him. In such a mechanism, we say that a player is {\em truthful} if she is an agent and she reports her location truthfully, or he is a mediator and he reports all of the locations of the agents below him truthfully (as received from his direct descendants).
While our results from the previous sections show the existence of competitive mechanisms that are two-sided incentive compatible for trees with a single layer of mediators, the next theorem proves that even for a much weaker solution concept, ex-post incentive compatibility, competitive mechanisms with multiple layers of mediators are impossible.

\begin{definition}
A direct revelation mechanism is {\em ex-post incentive compatible} if for every player, being truthful always maximizes the player's utility, assuming that all other players are truthful.
\end{definition}
We note that this solution concept is weaker than being two-sided IC (or the natural generalization of this notion to more levels) as it does not require any player to ever have a dominant strategy in an induced game with only part of the other players (his direct kins) being truthful. Unfortunately, it is impossible to construct competitive ex-post IC mechanisms for $s > 2$.

\begin{theorem}
\label{thm:no-ex-post-ic-hier}
No mechanism with a finite competitive ratio is ex-post IC for instances with $s > 2$.
\end{theorem}
\begin{proof}
Consider the following instance on the metric interval $[0, 2]$ with $s=3$.
\begin{compactitem}
	\item Level $3$ contains the center.
	\item Level $2$ contains a single mediator $C$, reporting to the center.
	\item Level $1$ contains two mediators, which we will call $A$ and $B$, both reporting to $C$.
	\item Level $0$ contains five agents:
	\begin{compactitem}
		\item Agents $a, b$ and $c$ are represented by mediator $A$, and located in $0, 0$ and $1$, respectively.
		\item Agents $d$ and $e$ are represented by mediator $B$, and are both located at $2$.
	\end{compactitem}
\end{compactitem}

Assume for the sake of contradiction that there exists a finite competitive ratio mechanism which is ex-post IC for the above instance. Observe that if the center gets a report stating that all five agents are located at a point $p$, then the mechanism must locate the facility at $p$ to have a bounded competitive ratio. 
Assuming all players are truthful, except maybe for $C$. Then the following observations hold:
\begin{compactitem}
	\item Let $S$ be the multi-set of the positions of the agents as reported to $C$. Then, $S$ represents the true locations of all the agents.
	\item The center observes the report of $C$.
\end{compactitem}
Let $m$ be $S$'s median. By Observation~\ref{ob:median_optimal}, $m$ is the optimal facility location for $C$. If $C$ deviates and reports that all agents are located at $m$, the center will have to locate the facility at $m$, due to the above discussion. Thus, for $C$ to have no incentive to deviate, the center must always locate the facility at the median of the locations it gets.

Consider now the situation that all players are truthful, except maybe for $A$. Let us considers $A$'s situation.
$A$ is reported, correctly, that his agents are located at $0, 0$ and $1$. Thus, the optimal facility location for $A$ is $0$. If $A$ reports truthfully, the center will get the reports $0, 0, 1, 2$ and $2$, and, by the above discussion, will locate the facility at $1$. On the other hand, if $A$ deviates, and reports that all his agents are located at $0$, the center will get the reports $0, 0, 0, 2$ and $2$, and will locate the facility at $0$ (which is better from $A$'s point of view). Thus, the mechanism considered is not ex-post IC.
\end{proof}

It seems that in any mechanism that satisfies ``minimal'' incentive properties, the only ``useful'' part of the information reported by a player to its ancestor is the optimal location (from the perspective of that player) with respect to that player's input. That is, the mechanism's output cannot change if the locations reported by a player are all replaced by the optimal location of that player (with respect to that player's input). This observation naturally suggests the following mechanism which generalizes the Weighted Median Mechanism.
The mechanism iteratively computes weighted medians (breaking ties in a report-independent way) for each mediator, bottom up,
and outputs the final location.
Consider the example presented in the proof of Theorem~\ref{thm:no-ex-post-ic-hier}. The suggested mechanism would behave as if $A$ reports that all his agents are located at $0$, and $C$ reports that all his agents are located at $0$ (as this is the median of his input which is $0, 0, 0, 2, 2$). The final outcome would be to locate the facility at $0$. Notice that this outcome, while not optimal, is not too far from optimality.

This raises few questions which we address next. First, assuming that all players are truthful, what is the approximation that this mechanism achieves? Second, what kind of incentive property does this  mechanism satisfy? And finally, for this incentive property, is there any other mechanism that is substantially better?

We begin by establishing some notation.
Let $d_{i, j}$ be the $j^{th}$ node of the $i^{th}$ level in the tree, where the level of the leaves that are furthest away from the root is $i = 0$ (notice that $d_{s, 1}$ is the root). For every $0 \leq i \leq s$, let $m_i$ denote the number of nodes appearing in level $i$ of the tree. For every mediator $d_{i, j}$ we denote by $C_{i, j}$ the set of children of $d_{i, j}$ in the tree and by $A_{i, j}$ the set of leaves (agents) that descent from $d_{i, j}$. For consistency, if $d_{i, j}$ is an agent, then $C_{i, j} = \varnothing$ and $A_{i, j} = \{d_{i, j}\}$. Finally, for every agent $a$, let $t_a$ denote the location of $a$.

We can now formally define the mechanism. We note that this is not a direct revelation mechanism, but rather, each player is asked to report a single location.

In the {\em Iterative Weighted Median Mechanism (IWMM)}:
\begin{compactitem}
	\item An agent reports her location.
	\item A mediator reports the weighted median of the reports he gets. More formally, let $\ell_{i, j}$ be the report of player $d_{i, j}$ and let $z_{i, j}$ be an arbitrary point selected independently of any reports. Then, $\ell_{i, j}$ is the weighted median of the multi-set $\{(\ell_{i - 1, j'}, |A_{i - 1, j'}|) ~|~ d_{i - 1, j'} \in C_{i, j}\}$ closest to $z_{i, j}$.
	\item The center locates the facility at the point it would have reported according to the above rule, if it were a mediator.
\end{compactitem}

Notice that IWMM is well defined by Observation~\ref{ob:closest_median_exists}. The next proposition summarizes the competitiveness of IWMM when viewed as an algorithm (with respect to its input).

\begin{proposition}\label{prop:IWMM-comp}
Assume that in Iterative Weighted Median Mechanism every agent is truthful and every mediator follows the protocol, then the mechanism has a competitive ratio of $2^s-1$ for any tree metric.
\end{proposition}
{
\begin{proof}
Let $F_{OPT}$ and $OPT$ denote an optimal location for the facility and the corresponding social cost, respectively. We prove by induction that for every level $0 \leq i \leq s$, $\sum_{j = 1}^{m_i} \sum_{a \in A_{i, j}} \dist(\ell_{i, j}, t_a) \leq (2^i - 1) \cdot OPT$, assuming all the players follow the mechanism. Notice that plugging $i = s$ implies the lemma since the center locates the facility at $\ell_{s, 1}$.

Let us start with the base of the induction. Every agent $d_{0, j}$ of level $0$ reports its location. Thus, $\ell_{0, j} = t_{d_{0, j}}$ and $A_{0, j} = \{d_{i, j}\}$. Using these observations we get:
\[
	\sum_{j = 1}^{m_0} \sum_{a \in A_{0, j}} \dist(\ell_{0, j}, t_a)
	=
	\sum_{j = 1}^{m_0} \dist(\ell_{0, j}, t_{d_{0, j}})
	=
	0
	=
	(2^0 - 1) \cdot OPT
	\enspace.
\]

This completes the proof of the induction base. Assume the claim holds for $i' \geq 0$, and let us prove it for $i = i' + 1$. Given an agent $a$ of level $i'$ or lower, let $j(i', a)$ be the (single) index of a player of level $i'$ for which  $a \in A_{i', j(i', a)}$. If $d_{i, j}$ is a mediator, then, by Observation~\ref{ob:median_optimal}:
\[
	\sum_{a \in A_{i, j}} \dist(\ell_{i, j}, \ell_{i - 1, j(i - 1, a)})
  \leq
  \sum_{a \in A_{i, j}} \dist(F_{OPT}, \ell_{i - 1, j(i - 1, a)})
  \enspace.
\]
Let $M_i$ be the set of mediators on level $i$, then summing the above inequality over all $j \in M_i$ gives:
\[
    \sum_{j \in M_i} \sum_{a \in A_{i, j}} \dist(\ell_{i, j}, \ell_{i - 1, j(i - 1, a)})
    \leq
    \sum_{j \in M_i} \sum_{a \in A_{i, j}} \dist(F_{OPT}, \ell_{i - 1, j(i - 1, a)})
    \enspace.
\]
Using the triangle inequality, we get:
{
\allowdisplaybreaks
\begin{align*}
		&
		\sum_{j \in M_i} \sum_{a \in A_{i, j}} \left[\dist(\ell_{i, j}, t_a) - \dist(t_a, \ell_{i - 1, j(i - 1, a)})\right]\\
		\leq{} &
    \sum_{j \in M_i} \sum_{a \in A_{i, j}} \dist(\ell_{i, j}, \ell_{i - 1, j(i - 1, a)})
		\leq
    \sum_{j \in M_i} \sum_{a \in A_{i, j}} \dist(F_{OPT}, \ell_{i - 1, j(i - 1, a)})\\
    \leq{} &
    \sum_{j \in M_i} \sum_{a \in A_{i, j}} \left[\dist(F_{OPT}, t_a) + \dist(t_a, \ell_{i - 1, j(i - 1, a)})\right]
    \enspace.
\end{align*}
}
Rearranging, gives:
\begin{align*}
	\sum_{j \in M_i} \sum_{a \in A_{i, j}} \dist(\ell_{i, j}, t_a)
	\leq{} &
	\sum_{j \in M_i} \sum_{a \in A_{i, j}} \left[\dist(F_{OPT}, t_a) + 2 \cdot \dist(t_a, \ell_{i - 1, j(i - 1, a)})\right]\\
	\leq{} &
	OPT + 2 \cdot \sum_{j \in M_i} \sum_{a \in A_{i, j}} \dist(t_a, \ell_{i - 1, j(i - 1, a)})\\
	={} &
	OPT + 2 \cdot \sum_{j = 1}^{m_{i - 1}} \sum_{a \in A_{i - 1, j}} \dist(t_a, \ell_{i - 1, j})\\
	\leq{} &
	OPT + 2 \cdot (2^{i - 1} - 1) \cdot OPT
	=
	(2^i - 1) \cdot OPT
	\enspace,
\end{align*}
where the first equality holds since both double sums sum over all the agents appearing below level $i$, and the third inequality is due to the induction hypothesis. The induction step follows since for every $j \not \in M_i$, $d_{i, j}$ is an agent, and thus, $\dist(\ell_{i, j}, t_{d_{i, j}}) =  0$, which implies:
\[
	\sum_{j = 1}^{m_i} \sum_{a \in A_{i, j}} \dist(\ell_{i, j}, t_a)
	=
	\sum_{j \in M_i} \sum_{a \in A_{i, j}} \dist(\ell_{i, j}, t_a)
	\leq
	(2^i - 1) \cdot OPT
	\enspace.
	\qedhere
\]
\end{proof}
}

By Theorem~\ref{thm:no-ex-post-ic-hier}, the direct revelation implementation of IWMM is not ex-post IC. Yet, this mechanism satisfies the following, much weaker, incentive property. Informally, every mediator, assuming that all his ascendants follow the protocol, and that the input he received from each of his direct descendants represents the true location of all agents of that direct descendant, will optimize his perceived utility by following the protocol.

\begin{definition}
A {\em single-location mechanism} is a non-direct revelation mechanism in which each player is reporting a single location to its direct ascendant.

A player of a single-location mechanism is {\em straightforward} if she is an agent and is truthful, or if he is a mediator and the location he reports is optimal with respect to his utility function assuming every agent represented by a direct descendant is located at the location reported by that direct descendent.

For a player in a single-location mechanism, we say that being straightforward is {\em naively optimal},
if being straightforward maximizes the above utility function under the assumption that his ascendants are straightforward

A single-location mechanism is {\em naively incentive compatible} if for every player being straightforward is naively optimal.
\end{definition}

We note that the term {\em naive} comes to emphasize that the players do not form Bayesian beliefs regarding the true locations of their agents and they do not try to optimize with respect to that belief, but rather (naively) assume that the reports they receive represent the true locations of the agents. Such a naive behavior is consistent with a mediator that never allows himself to harm his agents in the case when the reported locations are actually the true locations of his agents.

In Section~\ref{ssc:IWMM-IC}, we prove that IWMM is naively incentive compatible. Combining this with Proposition~\ref{prop:IWMM-comp} we derive the following theorem which summarizes the properties of IWMM.
\begin{theorem}
The Iterative Weighted Median Mechanism (IWMM) is a deterministic naively incentive compatible mechanism with a competitive ratio of $2^s-1$ for any tree metric.
\end{theorem}

Note that although our incentive property is extremely weak, the competitive ratio of IMWW degrades exponentially in $s$. Can such bad performance be avoided? We conclude this section with a lower bound showing that up to constant factors, it cannot.

\begin{theorem} \label{th:hardness_hierarchy}
Fix any constant $\ee > 0$ and tree depth $s\geq 3$.
Then, no mechanism (possibly randomized)
that is naively incentive compatible has a competitive ratio of $2^{s-2} - 1-\ee$ even for line metrics.
\end{theorem}

The proof of Theorem~\ref{th:hardness_hierarchy} is based on the following construction.

\begin{example}
\label{exmple:lower-bound-hierarchy}
Consider the metric on the interval $[0,1]$.
Fix $r$ to be an arbitrary positive integer, and fix a depth $s\geq 3$. We next describe a pair of instances. In both instances, the center has a single child mediator. Every mediator of level $i > 1$ represents $2$ mediators of level $i - 1$, and every mediator of level $1$ represents $r$ agents, except for $d_{1, 1}$ which represents $r + 1$ agents. Notice that there are in total $r \cdot 2^{s - 2} + 1$ agents in each instance of the pair.

The first instance, denoted by $I_1$, has all the agents located at point $0$. The second instance, denoted by $I_2$, has all the agents located at $1$, except for the agents of $d_{1, 1}$ which are located at $0$.
\end{example}

\begin{lemma}
Any (possibly randomized) mechanism
that is naively incentive compatible and has a finite competitive ratio must locate the facility at $0$ with probability $1$
given either $I_1$ or $I_2$.
\end{lemma}
\begin{proof}
The claim is trivial for $I_1$ since all the agents of $I_1$ are located at $0$. Notice that the center of $I_1$ has a single child mediator, and therefore, gets only a single report (which is equal to $0$). Let us prove that the center gets the same report also under $I_2$, and therefore, the center must behave in the same way for both instances.

Formally, we need to prove that $d_{s - 1, 1}$ reports $0$. We do so, by proving inductively that $d_{i, 1}$ reports $0$ for every $1 \leq i \leq s - 1$. Notice that all the agents of $d_{1, 1}$ are located at $0$, and therefore, $d_{1,1}$ must report $0$ to be straightforward. This completes the base of the induction. Assume $d_{i - 1, 1}$ reports $0$, and let us prove that so does $d_{i, 1}$ for $2 \leq i \leq s - 1$. Notice that $d_{i, 1}$ gets only two reports. One report from $d_{i - 1, 1}$ of value $0$ which represents $r \cdot 2^{i - 2} + 1$ agents, and another report which represents only $r \cdot 2^{i - 2}$ agents. Clearly, the optimal location for $d_{i, 1}$'s utility function (with respect the reports it gets) is $0$, because this is the reported location of the majority of the agents. Therefore, $d_{i, 1}$ must report $0$ to be straightforward.
\end{proof}

We are now ready to prove Theorem~\ref{th:hardness_hierarchy}.

\begin{proof}[Proof of Theorem~\ref{th:hardness_hierarchy}]
Fix a given naively IC mechanism with a bounded competitive ratio. The mechanism must locate the facility at $0$ with probability $1$ given $I_2$. However, there are only $r+1$ agents located at $0$ while there are $r \cdot (2^{s - 2} - 1)$ agents located at $1$. Thus, the mechanism has a cost of: $r \cdot (2^{s - 2} - 1)$, while the optimal facility location has a cost of no more than $r + 1$. The competitive ratio of the mechanism is, therefore, no better than:
\[
	\frac{r \cdot (2^{s-2} - 1)}{r+1}
	\enspace,
\]
which tends to $2^{s-2}-1$ as $r$ goes to infinity.
\end{proof}

\subsection{Proof: IWMM is Naively IC} \label{ssc:IWMM-IC}

In this section we prove that IWMM is Naively IC. This proof tends to follow the same ideas used in the analysis of WMM, but there are some subtle differences due to the hierarchical natural IWMM and the different solution concepts of the two algorithms (more specifically, WMM is a direct revelation mechanism, while IWMM is a single location mechanism). Let us start with a generalization of Observation~\ref{ob:point_moves} for hierarchical settings.

\begin{lemma} \label{le:change_multilevel}
Fix two mediators $d_{i, j}$ and $d_{i', j'}$ for which $d_{i', j'}$ is a descendent of $d_{i, j}$ in the tree. Assume $\ell_{i', j'}$ is changed to be $\hat{\ell}_{i', j'}$ by some change in the strategies of $d_{i', j'}$ and its descendents. Let $\hat{\ell}_{i, j}$ be the new location of $\ell_{i, j}$ after the change. If every direct ascendant of $d_{i', j'}$ is straightforward and $\ell_{i, j} \neq \hat{\ell}_{i,j}$, then $\ell_{i, j}$ and $\hat{\ell}_{i, j}$ are both on the path between $\ell_{i', j'}$ and $\hat{\ell}_{i', j'}$, and $\ell_{i', j'}$ is closer to $\ell_{i, j}$ than to $\hat{\ell}_{i, j}$.
\end{lemma}
\begin{proof}
We prove the lemma by induction on $i - i'$. For $i - i' = 0$ the lemma holds trivially, so
let us now assume the claim holds for $i - i' < h$, and let us prove it for $i - i' = h$. We also assume every ascendor of $d_{i', j'}$ is straightforward and $\ell_{i, j} \neq \hat{\ell}_{i,j}$ (otherwise, there is nothing to prove).

Let $d_{i-1, j''}$ be the (single) direct descendent of $d_{i, j}$ which is an ascendant of $d_{i', j'}$. Let $\hat{\ell}_{i-1, j''}$ be the new value of $\ell_{i-1, j''}$ after the change in $\ell_{i', j'}$. If $\ell_{i-1, j''} = \hat{\ell}_{i-1, j''}$, then we get $\ell_{i,j} = \hat{\ell}_{i, j}$ which contradicts our previous assumptions. Thus, we must have $\ell_{i-1, j''} \neq \hat{\ell}_{i-1, j''}$. By the induction hypothesis, this implies that $\ell_{i-1, j''}$ and $\hat{\ell}_{i-1, j''}$ are both on the path between $\ell_{i', j'}$ and $\hat{\ell}_{i', j'}$, and $\ell_{i - 1, j''}$ is the one closer to $\ell_{i', j'}$ among the two.

By Observation~\ref{ob:point_moves}, $\ell_{i, j}$ and $\hat{\ell}_{i, j}$ are both on the path between $\ell_{i - 1, j''}$ and $\hat{\ell}_{i - 1, j''}$, and $\ell_{i, j}$ is closer to $\ell_{i - 1, j''}$ than $\hat{\ell}_{i, j}$. Combining this with the previous observation, we get that the path from $\ell_{i', j'}$ to $\hat{\ell}_{i', j'}$ goes through the points $\ell_{i-1, j''}, \ell_{i, j}, \hat{\ell}_{i, j}$ and $\hat{\ell}_{i-1, j''}$, in that order, which completes the proof of the induction step.
\end{proof}

Using the above lemma, we can now prove that IWMM is naively IC.

\begin{lemma}
The Iterative Weighted Median Mechanism is naively incentive compatible.
\end{lemma}
\begin{proof}

Consider an arbitrary player $d_{i, j}$. We should prove that being straightforward is naively incentive compatible for $d_{i, j}$. In the rest of this proof let $\ell_{i, j}$ denote the report of $d_{i, j}$ assuming he follows the mechanism. Let $\hat{\ell}_{i, j} \neq \ell_{i, j}$ be an arbitrary possible deviation of $d_{i, j}$.
Our objective is to prove that assuming all of $d_{i, j}$ ascendants are straightforward, then reporting $\hat{\ell}_{i, j}$ is no better for $d_{i, j}$ than being straightforward, with respect to the utility function induced by the reports of $d_{i, j}$'s direct descendents.
Let $\ell_{s, 1}$ and $\hat{\ell}_{s, 1}$ be the locations picked by the center given that $d_{i, j}$ reports $\ell_{i, j}$ and $\hat{\ell}_{i, j}$, respectively. If $\ell_{s, 1} = \hat{\ell}_{s, 1}$ then the deviation is clearly non-beneficial. Otherwise, by Lemma~\ref{le:change_multilevel}, the path from $\ell_{i, j}$ to $\hat{\ell}_{i, j}$ goes through $\ell_{s, 1}$ and $\hat{\ell}_{s, 1}$ in that order.

For every $x \in [0, 1]$, let $p_x$ be the point along the path between $\ell_{i, j}$ and $\hat{\ell}_{i, j}$ obeying $\dist(\ell_{i, j}, p_x) = x \cdot \dist(\ell_{i, j}, \hat{\ell}_{i, j})$, and let $c_x$ be the value of the $d_{i, j}$'s objective function (with respect to the reports he gets from his direct descendents) if the facility is located at $p_x$. Let $p$ be an arbitrary point, and let us observe the change in $\dist(p_x, p)$ as $x$ increases. Clearly, $\dist(p_x, p)$ initially decreases at a constant rate and then starts to increase again in a constant rate (possibly the decrease part or increase part is of zero length).
Hence, the derivative of $c_x$ is a non-decreasing function of $x$ (except for a finite number of $x$ values for which it might not be defined). We know that for $x = 0$, $p_x$ is an optimal point for $d_{i, j}$. Hence, $c_x$ is minimized for $x = 0$. Thus, the derivative of $c_x$ is non-negative for every $x > 0$, whenever it is defined. Let $y$ and $y'$ be the values of $x$ for which $\ell_{s, 1} = p_y$ and $\hat{\ell}_{s, 1} = p_{y'}$. Then, we have $y < y'$, and therefore, $c_y \leq c_{y'}$, \ie, locating the facility at $\hat{\ell}_{s, 1}$ does not decrease the objective function of $d_{i, j}$ in comparison to locating it at $\ell_{s, 1}$. Hence, deviating to $\hat{\ell}_{i, j}$ is non-beneficial from $d_{i, j}$'s perspective.
\end{proof}
\section{Conclusion}

We studied the impact of strategic mediators on the competitive ratio of IC mechanisms in a facility location setting. Our results show that a single layer of mediation cause a moderate degradation in the competitive ratio, which becomes much worse as additional layers of mediation are introduced. We also showed that randomized mechanisms perform better than deterministic ones. 

Strategic mediators appear in many real world scenarios, and we believe it is important to study the implications of their behavior in various settings. For example, in display advertising, one common practice is for a mediator to buy advertisement space on a web page and split it between multiple advertisers he represents. Assume the mediator gets bids from potential advertisers, and based on these bids decides how to bid for the space (in the ad exchange auction). If he wins the space he also need to decide how to split the newly bought space among his advertisers. The mediator needs a strategy that will incentivize his advertisers to be truthful, which is not a trivial task even if the ad exchange uses a second price auction. Studying the effect on the social welfare and revenue of such mediators is an interesting open problem.

\bibliographystyle{plain}
\bibliography{bib}

\appendix
\section*{APPENDIX}
\section{Missing Proofs of Section~\ref{sec:model}} \label{ap:missing_model_proofs}

\subsection{Proof of Observation~\ref{ob:median_optimal}}

In this subsection we prove the following observation.

\begin{repobservation}{ob:median_optimal}
For every non empty finite multi-set $S$ of weighted points, there is at least one weighted median. Moreover, a point $p$ is a weighted median of $S$ if and only if locating the facility at $p$ minimizes the weighted total cost of a set of agents located at the points of $S$ (\ie, the sum $\sum_{\tilde{q} \in S} w(\tilde{q}) \cdot f(\tilde{q}, S) \cdot \dist(p, \tilde{q})$).
\end{repobservation}

The observation follows immediately from two lemmata. The first lemma shows that any multi-set $S$ of weighted points always have a weighted median. This lemma is a simple extension of a well-known result of Jordan~\cite{Jordan1869}.
\begin{lemma}
For every non empty finite multi-set $S$ of weighted points there is at least one weighted median.
\end{lemma}
\begin{proof}
Given a point $p$, let $\PP_p$ denote the partition of $S$ into maximal sub multi-sets such that every path between two weighted points of different sub multi-sets of $\PP_p$ goes through $p$. Observe that a point $p$ is a weighted median of $S$ if and only if all the sub multi-sets of $\PP_p$ are of weight at most $w(S) / 2$. Let $R = \{p ~|~ \exists_w (p, w) \in S\} \cup V$. Pick any point $r$ in $R$. If $r$ is a weighted median, we are done. If not, there is exactly one sub multi-set $P \in \PP_r$ with a weight larger than $w(S)/2$. By the definition of $R$ there must be a point
$r'\in R$ ($r' \neq r$) such that the path between every weighted point of $P$ and $r$ goes through $r'$. We claim that if $r'$ is not a weighted median, then there must be a sub multi-set $P' \in \PP_{r'}$ that is a subset of $P \setminus S_r$ and has weight larger than $w(S)/2$.
If this is so, we can replace $r$ by $r'$ and $P$ by $P'$ and continue in the same way till we find a weighted median (the sum $\sum_{\tilde{q} \in P} \dist(\tilde{q}, r)$ strictly decreases with each iteration, and can take only a finite set of values).

Next we show that if $r'$ is not a weighted median, then there must be a sub multi-set $P' \in \PP_{r'}$ that is a subset of $P$ and has weight larger than $w(S)/2$. Since $r'$ is not a weighted median, there exists a sub multi-set $P'' \in \PP_{r'}$ with weight over $w(S)/2$. If $P'' \setminus P \neq \varnothing$, then by the definition of $r'$ it holds that
 $P'' \cap P = \varnothing$. Therefore, $w(P'') \leq w(S) - w(P) < w(S) - w(S)/2 = w(S) / 2$, which is a contradiction. Thus, we must have $P'' \subseteq P$.
\end{proof}

The next lemma proves the second part of Observation~\ref{ob:median_optimal}, \ie, that the weighted median points are exactly the optimal facility locations for minimizing the weighted total cost of a set of agents located at the weighted points of $S$.

\begin{lemma}
Given a non empty finite multi-set $S$ of weighted points, a point $p$ is a weighted median of $S$ if and only if locating the facility at $p$ minimizes the weighted total cost of a set of agents located at the weighted points of $S$ (\ie, the sum $\sum_{\tilde{q} \in S} w(\tilde{q}) \cdot f(\tilde{q}, S) \cdot \dist(p, \tilde{q})$).
\end{lemma}
\begin{proof}
Consider two points $p$ and $q$, where $p$ is a median of $S$ (such a point exists due to the previous lemma). Assume the facility is located at $q$, and let us analyze the effect on the cost paid by the agents of $S$ if the facility is moved to $p$. Let $S' \subseteq S$ be the multi-set of weighted points of $S$ that are connected to $q$ via a path that does not go through $p$. By the definition a median, $w(S') \leq m_p \leq w(S) / 2$.

The path in $T$ from every weighted point of $S \setminus S'$ to $q$ goes through $p$. Hence, the move of the facility from $q$ to $p$ decreases the cost of agents located at these weighted points by $\dist(p, q)$. On the other hand, the move can increase the cost of an agent located at any other weighted point by at most $\dist(p, q)$. Hence, the change in the weighted total cost following the move of the facility is at most:
\begin{align} \label{eq:median_optimal_facility_move_change}
    w(S') \cdot \dist(p, q) - w(S \setminus S') \cdot \dist(p, q)
    ={} &
    (2w(S') - w(S)) \cdot \dist(p, q) \\ \nonumber
    \leq{} &
    (2 \cdot (w(S) / 2) - w(S)) \cdot \dist(p, q)
    =
    0
    \enspace.
\end{align}

Thus, we proved that locating the facility at a median point $p$ of $S$ induces as low a weighted total cost for the agents of $S$ as locating it in every other point. In other words, any median point $p$ of $S$ minimizes the weighted total cost incurred by a set of agents located at the weighted points of $S$. To prove the other direction, let us tighten the above analysis. Assume $q$ is not a median of $S$. Then, there exists a set $S'' \subseteq S$ such that the weighted points of $S''$ can be connected via paths that do not go through $q$, and $w(S'') > w(S) / 2$.

If $S'' \cap S' = \varnothing$, then $w(S') \leq w(S) - w(S'') < w(S) - w(S) / 2 = w(S)/2$. Plugging this into (\ref{eq:median_optimal_facility_move_change}) shows that the change in the total weighted cost paid by the agents of $S$ following the move of the facility from $q$ to $p$ is strictly negative. Otherwise, we note that $S'' \not \subseteq S'$ because $w(S'') > w(S) / 2 \geq w(S')$. Hence, there exist two weighted points $\tilde{q}' \in S'' \cap S'$ and $\tilde{q}'' \in S'' \setminus S'$. Since $\tilde{q}' \in S'$ and $\tilde{q}'' \not \in S'$, we learn that the path connecting $\tilde{q}'$ to $\tilde{q}''$ goes through $p$. However, this path cannot go through $q$ because both $\tilde{q}'$ and $\tilde{q}''$ are in $S''$. This implies that the additional cost of $\tilde{q}'$ from the move of the facility from $q$ to $p$ is lower than $\dist(p, q)$. However, in the analysis of (\ref{eq:median_optimal_facility_move_change}) we upper bounded this additional cost by $\dist(p, q)$, hence, $0$ is a strict upper bound on the change in the weighted total cost paid by the agents of $S$.

In conclusion, we proved that the assumption that $q$ is not a median of $S$ implies that the total weighted cost of a set of agents located at the weighted points of $S$ strictly decreases if the facility is moved from $q$ to $p$. Hence, any non-median point like $q$ does not minimize the total weighted cost of such a set of agents.
\end{proof}

\subsection{Proof of Observation~\ref{ob:agent_side_optimal}}
\begin{repobservation}{ob:agent_side_optimal}
There exists a deterministic, agent-side incentive compatible mechanism that is optimal ($1$-competitive).
\end{repobservation}
\begin{proof}
This proof uses some results from Section~\ref{sc:deterministic_algorithm}. We would like to stress that this does not introduce circular references between our proofs. Consider the following mechanism. Let $z$ be an arbitrary point chosen independently of the reports of the agents, and let $S$ be the multi-set of the reports of all agents (as reported by the mediators to the center). We consider the mechanism which chooses as the facility location the unique median of $S$ closest to $z$.

The existence of a unique median point closest to $z$ (\ie, closer to $z$ than any other median point) is guaranteed by Observation~\ref{ob:closest_median_exists}. The optimality of the above mechanism follows immediately from Observation~\ref{ob:median_optimal}. Thus, we are only left to show that the mechanism is agent-side incentive compatible.

Consider an arbitrary agent $a_{i, j}$. Assuming mediator $d_i$ is truthful, we would like to show that $a_{i, j}$ has no incentive to report a fake location $t'_{i, j} \neq t_{i, j}$. Let $F$ and $F'$ be the facility locations picked by the mechanism assuming $a_{i, j}$ reports $t_{i, j}$ and $t'_{i, j}$, respectively. Similarly, let $S$ and $S'$ be the multi-sets of the reports of all agents (as reported by the mediators to the center) assuming $a_{i, j}$ reports $t_{i, j}$ and $t'_{i, j}$, respectively. If $F = F'$, then reporting the fake location $t'_{i, j}$ is clearly non-beneficial for $a_{i, j}$. Otherwise, observe that since $d_i$ is truthful, $S' = S \cup \{t'_{i, j}\} \setminus \{t_{i, j}\} $. Hence, by Observation~\ref{ob:point_moves}, $F$ and $F'$ are both on the path connecting $t_{i, j}$ and $t'_{i, j}$, and $F$ is closer to $t_{i, j}$ than $F'$. Thus, $\dist(t_{i, j}, F) < \dist(t_{i, j}, F')$, \ie, it is strictly non-beneficial for $a_{i, j}$ to report the fake location $t'_{i, j}$.
\end{proof}

\subsection{Proof of Lemma~\ref{le:single_better_than_distribution}}

In this subsection we prove the following lemma.

\begin{replemma}{le:single_better_than_distribution}
For any distribution over points $F$ there exists a point $p(F)$ such that for any finite multi-set $S$ of points:
\[
    \underset{p' \sim F}{\mathbb{E}} \left[\sum_{q \in S} \dist(p', q) \cdot f(q, S)\right]
    \geq
    \sum_{q \in S} \dist(p(F), q) \cdot f(q, S)
    \enspace.
\]
Moreover, for the euclidian metric on $[a, b]$ (for arbitrary $a$ and $b$) the expected location according to $F$ can serve as such a point $p(F)$.
\end{replemma}

The proof of the lemma is done by defining a series of probability distributions $F_1, F_2, \ldots, F_\ell$ that have two properties:
\begin{compactitem}
    \item For every $1 \leq i < \ell$, \[\underset{p' \sim F_i}{\mathbb{E}} \left[\sum_{q \in S} \dist(p', q) \cdot f(q, S)\right] \geq \underset{p' \sim F_{i + 1}}{\mathbb{E}} \left[\sum_{q \in S} \dist(p', q) \cdot f(q, S)\right] \enspace.\]
    \item When sampling according to $F_\ell$ one always gets the same point (a point distribution).
\end{compactitem}

Let us begin by defining $F_1$. Map each node to an arbitrary edge hitting it. For every edge $e \in E$, let $A(e)$ denote the set of points that are located on $e$ or a node mapped to $e$. Observe that $\{A(e)\}_{e \in E}$ is a disjoint partition of the metric. The probability distribution $F_1$ is identical to $F$ except that a point $p$ chosen according to $F_1$ never belongs to a set $A(e)$ if $\Pr_{p' \sim F}[p' \in A(e)] = 0$. Formally, let $E(F) = \{e \in E ~|~ \Pr_{p' \sim F}[p' \in A(e)] > 0\}$. Given a set $S$ of points:
\[
    \Pr_{p' \sim F_1}[p' \in S]
    =
    \Pr_{p' \sim F}[p' \in S \cap \cup_{e \in E(F)} A(e)]
    \enspace.
\]

\begin{lemma} \label{le:F_1_equal_F}
$F_1$ is a probability distribution, and $\underset{p' \sim F}{\mathbb{E}} \left[\dist(p', q)\right] = \underset{p' \sim F_1}{\mathbb{E}} \left[\dist(p', q)\right]$ for every point $q$.
\end{lemma}
\begin{proof}
Clearly $\Pr_{p' \sim F_1}[p' \in S] \geq 0$ for every set $S$ of points. For any countable collection of disjoint sets $S_1, S_2, \ldots$, we have:
\begin{align*}
    \Pr_{p' \sim F_1}\left[p' \in \biguplus_{i = 1}^\infty S_i\right]
    ={} &
    \Pr_{p' \sim F}\left[p' \in \left(\biguplus_{i = 1}^\infty S_i\right) \cap \cup_{e \in E(F)} A(e)\right]\\
    ={} &
    \Pr_{p' \sim F}\left[p' \in \biguplus_{i = 1}^\infty \left(S_i \cap \cup_{e \in E(F)} A(e) \right) \right] \\
    ={} &
    \sum_{i = 1}^\infty \Pr_{p' \sim F}\left[p' \in S_i \cap \cup_{e \in E(F)} A(e) \right]
    =
    \sum_{i = 1}^\infty \Pr_{p' \sim F_1}\left[p' \in S_i \right]
    \enspace.
\end{align*}
Finally,
\begin{align*}
    \Pr_{p' \sim F_1}[p' \in \cup_{e \in E} A(e)]
    ={} &
    \sum_{e \in E} \Pr_{p' \sim F_1} [p' \in A(e)]
    =
    \sum_{e \in E(F)} \Pr_{p' \sim F} [p' \in A(e)] \\
    ={} &
    \sum_{e \in E} \Pr_{p' \sim F} [p' \in A(e)] - \sum_{e \not \in E(F)} \Pr_{p' \sim F} [p' \in A(e)]
    =
    1 - \sum_{e \not \in E(F)} 0
    =
    1
    \enspace.
\end{align*}

Observe that for every set $S$ of points, we have:
\begin{align*}
    \Pr_{p' \sim F_1}[p' \in S]
    ={} &
    \Pr_{p' \sim F}[p' \in S \cap \cup_{e \in E(F)} A(e)]
    =
    \Pr_{p' \sim F}[p' \in S] - \sum_{e \not \in E(F)} \Pr_{p' \sim F}[p' \in S \cap A(e)]\\
    ={} &
    \Pr_{p' \sim F}[p' \in S] - \sum_{e \not \in E(F)} 0
    =
    \Pr_{p' \sim F}[p' \in S]
    \enspace,
\end{align*}
and therefore, for any point $q$, $\underset{p' \sim F}{\mathbb{E}} \left[\dist(p', q)\right] = \underset{p' \sim F_1}{\mathbb{E}} \left[\dist(p', q)\right]$.
\end{proof}

Before presenting the other probability distributions ($F_2, F_3, \ldots F_\ell$), we need the following notation.
Given a probability distribution $F'$, and a path $P$ in the tree $T$ such that $P \cap E(F') \neq \varnothing$, let $A(P) = \cup_{e \in P} A(e)$ and let $F'|_P$ be a probability distribution in which the probability mass of the path $P$ is concentrated at the average of this path. More formally, let $u,v$ be the end nodes of $P$, then for every $x \in [0, 1]$, we denote by $u(x)$ the point along $P$ obeying $\dist(u, u(x)) = x \cdot \dist(u, v)$. For a point $p' \in A(P)$ we denote by $x(p')$ the value for which $p' = u(x)$. Also, let $\mathbb{E}_{F'}[P] = u\left(\mathbb{E}_{p' \sim F'}[x(p') ~|~ p' \in A(P)]\right)$, then $F'|_P$ is a probability distribution which is identical to $F'$ except that it outputs $\mathbb{E}_{F'}[P]$ whenever $F'$ outputs a point along the path $P$.

\begin{lemma} \label{le:path_expecation_imporve}
Assume $P$ is a path in the tree and $F'$ is a probability distribution in which a point $p' \sim F'$ can belong to $A(e)$ if and only if $e \in E(F')$. Then $F'|_P$ is a probability distribution, and $\underset{p' \sim F'}{\mathbb{E}} \left[\dist(p', q)\right] \geq \underset{p' \sim F'|_P}{\mathbb{E}} \left[\dist(p', q)\right]$ for every point $q$.
\end{lemma}
\begin{proof}
Given a set $S$ of points:
\[
    \Pr_{p' \sim F'|_P}[p' \in S]
    =
    \Pr_{p' \sim F'}[p' \in S \setminus A(P)] + 1_{\mathbb{E}_{F'}[P] \in S} \cdot \Pr_{p' \sim F'}[p' \in A(P)]
    \enspace,
\]
where $1_{\mathbb{E}_{F'}[P] \in S}$ is an indicator taking the value $1$ when $\mathbb{E}_{F'}[P] \in S$, and the value $0$ otherwise. Lets prove that $F'|_P$ is a probability distribution. It is clear that $\Pr_{p' \sim F'|_P}[p' \in S] \geq 0$ for every set $S$ of points. Consider an arbitrary countable collection of disjoint sets $S_1, S_2, \ldots$. If $\mathbb{E}_{F'}[P]$ does not belong to any of the sets $S_i$, then:
\begin{align*}
    \Pr_{p' \sim F'|_P}\left[p' \in \biguplus_{i = 1}^\infty S_i\right]
    ={} &
    \Pr_{p' \sim F'}\left[p' \in \left(\biguplus_{i = 1}^\infty S_i\right) \setminus A(P)\right]
    =
    \Pr_{p' \sim F'}\left[p' \in \biguplus_{i = 1}^\infty \left(S_i \setminus A(P)\right)\right]\\
    ={} &
    \sum_{i = 1}^\infty \Pr_{p' \sim F'}\left[p' \in S_i \setminus A(P)\right]
    =
    \sum_{i = 1}^\infty \Pr_{p' \sim F'|_P}\left[p' \in S_i\right]
    \enspace.
\end{align*}
On the other hand, if the point $\mathbb{E}_{F'}[P]$ belongs to one of the sets $S_1, S_2, \ldots$, then it belongs to exactly one of them, and therefore:
\begin{align*}
    \Pr_{p' \sim F'|_P}\left[p' \in \biguplus_{i = 1}^\infty S_i\right]
    &=
    \Pr_{p' \sim F'}\left[p' \in \left(\biguplus_{i = 1}^\infty S_i\right) \setminus A(P)\right] + \Pr_{p' \sim F'}[p' \in A(P)]\\
    ={} &
    \Pr_{p' \sim F'}\left[p' \in \biguplus_{i = 1}^\infty \left(S_i \setminus A(P)\right)\right] + \Pr_{p' \sim F'}[p' \in A(P)]\\
    ={} &
    \sum_{i = 1}^\infty \Pr_{p' \sim F'}\left[p' \in S_i \setminus A(P)\right] + \Pr_{p' \sim F'}[p' \in A(P)]
    =
    \sum_{i = 1}^\infty \Pr_{p' \sim F'|_P}\left[p' \in S_i\right]
    \enspace.
\end{align*}
Finally, notice that $\mathbb{E}_{F'}[P] \in \cup_{e \in E} A(e)$. Therefore,
\begin{align*}
    \Pr_{p' \sim F'|_P}[p' \in \cup_{e \in E} A(e)]
    ={} &
    \Pr_{p' \sim F'|_P}\left[p' \in \left(\cup_{e \in E} A(e)\right) \setminus A(P)\right] + \Pr_{p' \sim F'}[p' \in A(P)]\\
    ={} &
    \Pr_{p' \sim F'|_P}\left[p' \in \cup_{e \in E} A(e) \right] - \Pr_{p' \sim F'}[p' \in A(P)] + \Pr_{p' \sim F'}[p' \in A(P)]\\
    ={} &
    \Pr_{p' \sim F'|_P}\left[p' \in \cup_{e \in E} A(e) \right]
    =
    1
    \enspace.
\end{align*}

Consider an arbitrary point $q$, and let $q'$ be the point on $P$ closest to $q$. Then,
\begin{align*}
    \underset{p' \sim F'}{\mathbb{E}} \left[\dist(p', q) ~|~ p' \in A(P)\right] \hspace{-1.3in} & \hspace{1.3in}
    =
    \dist(q, q') + \underset{p' \sim F'}{\mathbb{E}} \left[\dist(p', q') ~|~ p' \in A(P)\right]\\
    ={} &
    \dist(q, q') + \underset{p' \sim F'}{\mathbb{E}} \left[|x(p') - x(q')| ~|~ p' \in A(P)\right] \\
    ={} &
    \dist(q, q') + \underset{p' \sim F'}{\mathbb{E}} \left[\max\{x(p') - x(q'), x(q') - x(p')\} ~|~ p' \in A(P)\right]\\
    \geq{} &
    \dist(q, q') + \max \left\{\underset{p' \sim F'}{\mathbb{E}}
     \left[x(p') - x(q') ~|~ p' \in A(P)\right], \underset{p' \sim F'}{\mathbb{E}} \left[x(q') - x(p') ~|~ p' \in A(P)\right] \right\}\\
    ={} &
    \dist(q, q') + \max \left\{\underset{p' \sim F'}{\mathbb{E}} \left[x(p') ~|~ p' \in A(P)\right] - x(q'), x(q') - \underset{p' \sim F'}{\mathbb{E}} \left[x(p') ~|~ p' \in A(P)\right] \right\}\\
    ={} &
    \dist(q, q') + \max \left\{x(\mathbb{E}_{F'}[P]) - x(q'), x(q') - x(\mathbb{E}_{F'}[P]) \right\}\\
    ={} &
    \dist(q, q') + \dist(\mathbb{E}_{F'}[P], q')
    =
    \dist(\mathbb{E}_{F'}[P], q)
    \enspace.
\end{align*}
If $E(F') \subseteq P$, then:
\begin{align*}
    \underset{p' \sim F'}{\mathbb{E}} \left[\dist(p', q)\right]
    ={} &
    \underset{p' \sim F'}{\mathbb{E}} \left[\dist(p', q) ~|~ p' \in A(P)\right]
    \geq
    \dist(\mathbb{E}_{F'}[P], q) \\
    ={} &
    \underset{p' \sim F'|_P}{\mathbb{E}} \left[\dist(p', q) ~|~ p' \in A(P)\right]
    =
    \underset{p' \sim F'|_P}{\mathbb{E}} \left[\dist(p', q)\right]
    \enspace.
\end{align*}
Otherwise,
\begin{align*}
    \underset{p' \sim F'}{\mathbb{E}} \left[\dist(p', q)\right]
    ={} &
    \underset{p' \sim F'}{\mathbb{E}} \left[\dist(p', q) ~|~ p' \in A(P)\right] \cdot \Pr_{p' \sim F'}[p' \in A(P)]\\
    &+ \underset{p' \sim F'}{\mathbb{E}} \left[\dist(p', q) ~|~ p' \not \in A(P)\right] \cdot \Pr_{p' \sim F'}[p' \not \in A(P)]\\
    \geq{} &
    \dist(\mathbb{E}_{F'}[P], q) \cdot \Pr_{p' \sim F'}[p' \in A(P)]\\
    &+ \underset{p' \sim F'}{\mathbb{E}} \left[\dist(p', q) ~|~ p' \not \in A(P)\right] \cdot \Pr_{p' \sim F'}[p' \not \in A(P)]\\
    ={} &
    \underset{p' \sim F'|_P}{\mathbb{E}} \left[\dist(p', q) ~|~ p' \in A(P)\right] \cdot \Pr_{p' \sim F'}[p' \in A(P)]\\
    &+ \underset{p' \sim F'|_P}{\mathbb{E}} \left[\dist(p', q) ~|~ p' \not \in A(P)\right] \cdot \Pr_{p' \sim F'}[p' \not \in A(P)]\\
    ={} &
    \underset{p' \sim F'|_P}{\mathbb{E}} \left[\dist(p', q)\right]
    \enspace.
    \qedhere
\end{align*}
\end{proof}

Using the machinery developed above, we can now define the following process which constructs the distributions $F_2, F_2, \ldots, F_\ell$. For a probability distribution $F'$, let $Q(F')$ denote the set of points that can result from $F'$.

\begin{tabbing}
    1. Let $i \leftarrow 1$. \\
    2. If \= $E(F_i) > 1$, then: \\
    3. \> Let $P_i$ be a path containing at least two edges of $E(F_i)$.\\
    4. \> Let $F_{i + 1} \leftarrow F_i|_{P_i}$.\\
    5. Else if $|Q(F')| > 1$:\\
    6. \> Let $P_i$ be the path containing only the single edge of $E(F_i)$.\\
    7. \> Let $F_{i + 1} \leftarrow F_i|_{P_i}$.\\
    8. Else:\\
    9. \> Terminate the process.\\
    10. $i \leftarrow i + 1$.\\
    11. Return to line $2$.
\end{tabbing}

\begin{observation}
The above process always terminates.
\end{observation}
\begin{proof}
Observe that $|E(F_i)| > |E(F_{i + 1})|$ unless $|E(F_{i + 1})| = 1$ (in which case the process terminates after producing $F_{i + 1}$).
\end{proof}

Let $\ell$ be the index of the last probability distribution produced by the process, and let $p(F)$ be the single point that can be outputted by $F_\ell$. Let $q$ be an arbitrary point. From Lemmata~\ref{le:F_1_equal_F} and \ref{le:path_expecation_imporve}, we get:
\begin{align*}
    \underset{p' \sim F}{\mathbb{E}} \left[\dist(p', q)\right]
    =
    \underset{p' \sim F_1}{\mathbb{E}} \left[\dist(p', q)\right]
    \geq
    \underset{p' \sim F_2}{\mathbb{E}} \left[\dist(p', q)\right]
    \geq
    \ldots
    \geq{} &
    \underset{p' \sim F_\ell}{\mathbb{E}} \left[\dist(p', q)\right]\\
    ={} &
    \dist(p(F), q)
    \enspace.
\end{align*}

Consider now an arbitrary finite multi-set $S$ of points. Then, by the linearity of the expectation:
\[
    \underset{p' \sim F}{\mathbb{E}} \left[\sum_{q \in S} \dist(p', q) \cdot f(q, S)\right]
    =
    \sum_{q \in S} \underset{p' \sim F}{\mathbb{E}} \left[\dist(p', q) \right] \cdot f(q, S)
    \geq
    \sum_{q \in S} \dist(p(F), q) \cdot f(q, S)
    \enspace,
\]
which completes the proof of the first part of the Lemma~\ref{le:single_better_than_distribution}. Observe that if the metric is the euclidian metric on $[a, b]$, then there are two cases.
\begin{compactitem}
    \item $Q(F) = \{q\}$ for some point $q$. In this case $\ell = 1$, and $p(F) = q = \underset{p' \sim F}{\mathbb{E}} [p']$.
    \item $|Q(F)| > 1$. In this case, one possible execution of the above process chooses $P_1$ which is equal to the entire interval $[a, b]$. For this execution, we get $\ell = 2$ and $p(F) = \underset{p' \sim F_1}{\mathbb{E}} [p'] = \underset{p' \sim F}{\mathbb{E}} [p']$.
\end{compactitem}
Thus, in both cases $\underset{p' \sim F}{\mathbb{E}} [p']$ can serve as $p(F)$.

{
	\allowdisplaybreaks
	\section{Analysis of TRM} \label{app:rand_proofs}

In this appendix we analyze TRM and prove the following theorem.

\begin{reptheorem}{th:general_metrics_mechanism}
For any tree metric, the \emph{Tree Randomized Mechanism} (TRM) is a randomized two-sided IC mechanism with a competitive ratio of $2$.
\end{reptheorem}

The next two lemmata prove that $p(u)$ defines a legal probability distribution. Subsections~\ref{ssc:approximation_general} and \ref{th:truthfulness_general} prove the competitive ratio and the two-sided IC of TRM, respectively.

\begin{lemma}
For every node $u \in X$, $p(u) \in [0, 1]$.
\end{lemma}
\begin{proof}
First, let us consider a node $u \in X \setminus \{r\}$. For such a node $c(u) = n/4$ and $n/2 \geq \treesize(u) \geq n/4$ because $r$ is a weighted median of $L$. Hence, by the definition of $X$:
\[
    \treesize(u) - \sum_{u' \in X \cap \children(u)} \mspace{-36mu} (\treesize(u') - n/4)
    \leq
    \treesize(u)
    \leq
    n/2
    \enspace,
\]
which implies:
\[
    p(u)
    =
    \frac{\treesize(u) - c(u)}{n/2} - \frac{\sum_{u' \in X \cap \children(u)}(\treesize(u') - c(u'))}{n/2}
    \leq
    \frac{n/2 - n/4}{n/2}
    =
    1/2
    \enspace.
\]

Also, $\treesize(u) - \sum_{p' \in X \cap \children(u)}(\treesize(p') - n/4) \geq n/4$ because, either $|X \cap \children(u)| = 0$ or:
\begin{align*}
    &
    \treesize(u) - \sum_{u' \in X \cap \children(u)} \mspace{-36mu} (\treesize(u') - n/4)
    \geq
    n/4 + \treesize(u) - \sum_{u' \in X \cap \children(u)} \mspace{-36mu} \treesize(u')\\
    \geq{} &
    n/4 + \treesize(u) - \sum_{u' \in \children(p)} \mspace{-18mu} \treesize(u')
    =
    n/4 + \size(u)
    \geq
    n/4
    \enspace.
\end{align*}
Thus,
\[
    p(u)
    =
    \frac{\treesize(u) - c(u)}{n/2} - \frac{\sum_{u' \in X \cap \children(u)}(\treesize(u') - c(u'))}{n/2}
    \geq
    \frac{n/4 - n/4}{n/2}
    =
    0
    \enspace.
\]

We are now left to prove that $0 \leq p(r) \leq 1$. Clearly, $c(r) = n/2$ and $\treesize(r) = n$. By the definition of $X$ we get:
\[
    \treesize(r) - \sum_{u' \in X \cap \children(r)} \mspace{-36mu} (\treesize(u') - n/4)
    \leq
    \treesize(r)
    =
    n
    \enspace,
\]
which implies:
\[
    p(r)
    =
    \frac{\treesize(r) - c(r)}{n/2} - \frac{\sum_{u' \in X \cap \children(r)}(\treesize(u') - c(u'))}{n/2}
    \leq
    \frac{n - n/2}{n/2}
    =
    1
    \enspace.
\]

Also, $\treesize(r) - \sum_{u' \in X \cap \children(r)}(\treesize(u') - n/4) \geq n/2$ because:
\begin{compactitem}
    \item If $|X \cap \children(r)| \leq 1$:
    \begin{align*}
        \treesize(r) - \sum_{u' \in X \cap \children(r)} \mspace{-36mu} (\treesize(u') - n/4)
        \geq{} &
        \treesize(r) - \sum_{u' \in X \cap \children(r)} \mspace{-18mu} \treesize(u')\\
        \geq{} &
        \treesize(r) - |X \cap \children(r)| \cdot n/2
        \geq
        n/2
        \enspace.
    \end{align*}
    \item If $|X \cap \children(r)| \geq 2$:
    \begin{align*}
        &
        \treesize(r) - \sum_{u' \in X \cap \children(r)} \mspace{-36mu} (\treesize(u') - n/4)
        \geq
        n/2 + \treesize(r) - \sum_{u' \in X \cap \children(r)} \mspace{-36mu} \treesize(u')\\
        \geq{} &
        n/2 + \treesize(r) - \sum_{u' \in \children(r)} \mspace{-18mu} \treesize(u')
        =
        n/2 + \size(r)
        \geq
        n/2
        \enspace.
    \end{align*}
\end{compactitem}

Thus,
\[
    p(r)
    =
    \frac{\treesize(r) - c(r)}{n/2} - \frac{\sum_{u' \in X \cap \children(r)}(\treesize(u') - c(u'))}{n/2}
    \geq
    \frac{n/2 - n/2}{n/2}
    =
    0
    \enspace.
    \qedhere
\]
\end{proof}

\begin{lemma} \label{le:tree_probability}
For every node $u \in X$, $\sum_{u' \in X \cap \subtree(u)} p(u') = \frac{\treesize(u) - c(u)}{n/2}$, hence, for $u = r$, we get $\sum_{u' \in X} p(u') = 1$.
\end{lemma}
\begin{proof}
We prove the lemma by induction on the size of $|X \cap \subtree(u)|$. If $|X \cap \subtree(u)| = 1$, then $X \cap \children(u) = \varnothing$, and therefore,
\[
    p(u)
    =
    \frac{\treesize(u) - c(u)}{n/2} - \frac{\sum_{u' \in X \cap \children(u)}(\treesize(u') - c(u))}{n/2}
    =
    \frac{\treesize(u) - c(u)}{n/2}
    \enspace.
\]
Assume now that the claim hold for $|X \cap \subtree(u)| < h$, and let us prove it for $|X \cap \subtree(u)| = h$. Every point $u' \in X \cap \children(u)$ obeys: $X \cap \subtree(u') \subseteq X \cap \subtree(u) \setminus \{u\}$, and therefore, $|X \cap \subtree(u')| < |X \cap \subtree(u)|$, \ie, the induction hypothesis holds for $u'$. Thus,
\begin{align*}
    \sum_{u' \in X \cap \subtree(u)} p(u')
    ={} &
    p(u) + \sum_{\substack{u' \in X \cap \subtree(u) \\ u' \neq u}} \mspace{-27mu} p(u') \\
    = {} &
    \frac{\treesize(u) - c(u)}{n/2} - \frac{\sum_{u' \in X \cap \children(u)}(\treesize(u') - c(u'))}{n/2}\\
     & + \sum_{u' \in X \cap \children(u)} \frac{\treesize(u') - c(u')}{n/2}\\
    ={} &
    \frac{\treesize(u) - c(u)}{n/2}
    \enspace.
		\qedhere
\end{align*}
\end{proof}

\subsection{Proof: TRM is \texorpdfstring{$2$}{2}-Competitive} \label{ssc:approximation_general}

In this subsection we prove that TRM is $2$-competitive as an algorithm (when players are truthful). Let $V' = V \cup \{t_{i, j} | 1 \leq i \leq k \wedge 1 \leq j \leq n_i\}$, where $V$ is the set of nodes of the tree produced by the first part of TRM. Intuitively, $V'$ contains $V$ and the \emph{real} locations of all agents. The definition of $V'$ calls for the following extensions of $\treesize$, $\subtree$, $X$ and $c$:
\begin{compactitem}
    \item For a point $q$ outside of $V$, let $u_q$ be the first node of $V$ encountered when moving from $q$ away from $r$. Then, $\treesize(q) = \treesize(u_q)$ and $\subtree(q) = \subtree(u_q)$.
    \item Let $X' = \{q \in V' ~|~ \treesize(q) \geq n/4\}$.
    \item For every point $q \in X' \setminus X$, let $c(q) = n/4$.
\end{compactitem}

We also define $S$ as the set of segments connecting pairs of points $q$ and $q'$ of $V'$ that are not separated by any other point of $V'$. Observe that since $V \subseteq V'$, the interiors of the sections of $S$ are disjoint.

\begin{observation} \label{ob:optimal_in_u_prime}
There exists an optimal facility location $F$ in $V'$.
\end{observation}
\begin{proof}
Observe that every median of the multi-set $A = \{t_{i, j} | 1 \leq i \leq k \wedge 1 \leq j \leq n_i\}$ is an optimal facility location, and let $F$ be an arbitrary median of $A$. If $F \in V'$, then we are done. Otherwise, we also have $F \not \in A$. Let $F'$ be a point of $V'$ which is closest to $F$ among all points of $V'$. Consider an arbitrary multi-set $S'$ of points from $A \setminus \{F'\}$ that are connected by paths that do not pass through $F'$. By the definition of $F'$, every path between points of $A$ that passes through $F$ also passes through $F'$. Thus, the points of $S'$ are connected by paths that do not pass through $F$. Hence, $|S'| \leq |A| / 2$ because $F$ is a median of $A$. Since $S'$ is an arbitrary multi-set of points from $A \setminus \{F'\}$ that are connected by paths that do not pass through $F'$, $F'$ is a median of $A$ that belongs to $V'$.
\end{proof}

Observation~\ref{ob:optimal_in_u_prime} implies that for every agent $a_{i, j}$ the path between $F$ and $t_{i, j}$ is a
union of segments from $S$ with disjoint interiors.
It can be easily seen that the same is true also for each facility location that TRM might select. Let us focus on an arbitrary segment $s \in S$ connecting two points $q, q' \in V'$, where $q'$ is the point closer to $r$ in the pair (there must be such a point since $V \subseteq V'$). We say that an agent $a_{i,j}$ uses $s$ under the optimal solution if the path from $t_{i, j}$ to $F$ goes through $s$ (by Observation~\ref{ob:optimal_in_u_prime} we know that the path must either go through all of $s$ or though no part of it). Similarly, the probability that $a_{i, j}$ uses $s$ under TRM is the probability that the path from $a_{i, j}$ to the facility location chosen by TRM goes through $s$. Observe that in order to prove the approximation ratio of TRM, it is enough to prove that the number of agents using $s$ under the optimal solution is at least half of the expected number of agents using $s$ under TRM. Let $\overline{\treesize}(s)$ denote the number of agents whose location is in the subtree of $q$. More formally,
\[
    \overline{\treesize}(s)
    =
    |\{a_{i, j} ~|~ \mbox{the path from $t_{i, j}$ to $r$ goes through $q$}\}|
    \enspace.
\]

The rest of this subsection is split between two cases: $\overline{\treesize}(s) \leq n/2$ and $\overline{\treesize}(s) \geq n/2$.

\subsubsection{The case of \texorpdfstring{$\overline{\treesize}(s) \leq n/2$}{treesize(s) <= n/2}}

Let us first consider the case $\overline{\treesize}(s) \leq n/2$.

\begin{lemma} \label{le:twice}
$\treesize(q) \leq 2 \cdot \overline{\treesize}(s)$.
\end{lemma}
\begin{proof}
For an agent to be counted by $\treesize(q)$, it must belong to a mediator $d_i$ obeying $\ell_i \in \subtree(u_q)$. Since $\ell_i$ is the median of the (reported) locations of the agents, $\ell_i \in \subtree(u_q)$ implies that least half of the agent in $A_i$ are located in the subtree rooted at $q$ (because all the agents of $A_i$ located outside of this subtree can be connected via paths that does not go through $\ell_i$). Thus, whenever the agents of $A_i$ are counted by $\treesize(q)$, at least half of the agents in $A_i$ are also counted by $\overline{\treesize}(s)$.
\end{proof}

We can now upper bound the expected number of agents using $s$ under TRM.

\begin{observation} \label{ob:tree_probability_extended}
Lemma~\ref{le:tree_probability} holds for every point $s \in X'$.
\end{observation}
\begin{proof}
If $s \in X$ there is nothing to prove. Otherwise, recall that $u_s$ is the first node of $V$ encountered when moving from $s$ away from $r$. Then,
\[
    \sum_{u' \in \subtree(s)} p(u')
    =
    \sum_{u' \in \subtree(u_s)} p(u')
    =
    \frac{\treesize(u_s) - c(u_s)}{n/2}
    =
    \frac{\treesize(s) - c(s)}{n/2}
    \enspace,
\]
where the last equality holds because neither $s$ nor $u_s$ can be $r$.
\end{proof}

\begin{lemma} \label{le:TRM_pays}
The expected number of agents using $s$ under TRM is at most: $\max\{6 \cdot \overline{\treesize}(s) - 8n^{-1} \cdot (\overline{\treesize}(s))^2 - n/2, \overline{\treesize}(s)\}$.
\end{lemma}
\begin{proof}
Clearly $q \neq r$, and therefore, $c(q) = n/4$. We have to distinguish between two cases. The first case is $q \in X'$. In this case, by Observation~\ref{ob:tree_probability_extended}, the probability that TRM will locate the facility below $s$ in the tree is $(\treesize(q) - c(q)) / (n/2)$. If the facility is indeed located below $s$, then the number of agents that will use $s$ under TRM is $n - \overline{\treesize}(s)$. However, if the facility is located above $s$, then $\overline{\treesize}(s)$ agents will use $s$. In conclusion, the expected number of agents that use $s$ under TRM is:
\begin{align*}
    &
    \frac{\treesize(q)-c(q)}{n/2} \cdot \left(n - \overline{\treesize}(s)\right) + \left(1 - \frac{\treesize(q)-c(q)}{n/2}\right) \cdot \overline{\treesize}(s) \\
    ={} &
    \frac{\treesize(q)-c(q)}{n/2} \cdot \left(n - 2 \cdot \overline{\treesize}(s)\right) + \overline{\treesize}(s) \\
    \leq{} &
    \frac{2 \cdot \overline{\treesize}(s)-c(q)}{n/2} \cdot \left(n - 2 \cdot \overline{\treesize}(s)\right) + \overline{\treesize(s)}\\
    ={} &
    \frac{\left(2 \cdot \overline{\treesize}(s)-n/4\right) \cdot \left(n - 2 \cdot \overline{\treesize}(s)\right)  + \overline{\treesize(s)} \cdot (n/2)}{n/2}\\
    ={} &
    \frac{3n \cdot \overline{\treesize(s)} - 4 \cdot (\overline{\treesize}(s))^2 - n^2/4}{n/2}
    =
    6 \cdot \overline{\treesize(s)} - 8n^{-1} \cdot (\overline{\treesize}(s))^2 - n/2
    \enspace.
\end{align*}

Let us now consider the second case, \ie, the case that $q \not \in X'$. Notice that in this case $\treesize(q) < n/4$, and therefore, no point of $X'$ is located in the subtree rooted at $q$. This implies that the probability that the facility will be located below $q$ is $0$, \ie, exactly $\overline{\treesize}(s)$ agents use $s$ under TRM.
\end{proof}

On the other hand, we can also characterize the number of agents using $s$ under the optimal solution.

\begin{lemma} \label{le:opt_pays}
The number of agents using $s$ under optimal solution is: $\overline{\treesize}(s)$.
\end{lemma}
\begin{proof}
If $F$ is above $s$ in the tree, then the number of agents using $s$ under the optimal solution is clearly $\overline{\treesize}(s)$.

Otherwise, $\overline{\treesize}(s) \geq n/2$ because $F$ is a median of $A$. Hence, the number of agents using $s$ under the optimal solution is $n - \overline{\treesize}(s) \leq n/2$. Recall that we assumed $\overline{\treesize}(s) \leq n/2$, thus, we must have: $\overline{\treesize}(s) = n/2$, which implies: $n - \overline{\treesize}(s) = \overline{\treesize}(s) = n/2$.
\end{proof}

We can now conclude the proof for the case $\overline{\treesize}(s) \leq n/2$.

\begin{corollary} \label{co:segment_ratio}
The number of agents using $s$ under TRM is at most twice the number of agents using $s$ under the optimal solution.
\end{corollary}
\begin{proof}
If $\overline{\treesize}(s)$ agents or less use $s$, in expectation, under TRM, then there is nothing to prove. Otherwise, by Lemmata~\ref{le:TRM_pays} and \ref{le:opt_pays}, the ratio between the number of agents using $s$ under TRM and under the optimal solution is at most:
\[
    \frac{6 \cdot \overline{\treesize(s)} - 8n^{-1} \cdot (\overline{\treesize}(s))^2 - n/2}{\overline{\treesize}(s)}
    =
    6 - 8n^{-1} \cdot \overline{\treesize}(s) - \frac{n}{2 \cdot \overline{\treesize}(s)}
    \enspace.
\]
It can be easily checked that the last expression is maximized for $\overline{\treesize}(s) = n/4$, and the value of this maximum is:
\[
    6 - 8n^{-1} \cdot (n/4) - \frac{n}{2 \cdot (n/4)}
    =
    6 - 2 - 2
    =
    2
    \enspace.
    \qedhere
\]
\end{proof}

\subsubsection{The case of \texorpdfstring{$\overline{\treesize}(s) \geq n/2$}{treesize(s) >= n/2}}

Let us now consider the case $\overline{\treesize}(s) \geq n/2$.

\begin{lemma} \label{le:twice_case_2}
$n - \treesize(q) \leq 2 \cdot (n - \overline{\treesize}(s))$.
\end{lemma}
\begin{proof}
The proof of this lemma is analogous to the proof of Lemma~\ref{le:twice}.
\end{proof}

We can now upper bound the expected number of agents using $s$ under TRM.

\begin{lemma} \label{le:TRM_pays_case_2}
If $\overline{\treesize}(s) \geq 5n/8$, then the expected number of agents using $s$ under TRM is upper bounded by: $10 \cdot \overline{\treesize}(s) - 8n^{-1} \cdot (\overline{\treesize}(s))^2 - 5n/2$. Otherwise, it is upper bounded by $\max\{10 \cdot \overline{\treesize}(s) - 8n^{-1} \cdot (\overline{\treesize}(s))^2 - 5n/2, \overline{\treesize}(s)\}$.
\end{lemma}
\begin{proof}
Clearly $q \neq r$, and therefore, $c(q) = n/4$. We have to distinguish between two cases. The first case is that $q \in X'$. In this case, by Observation~\ref{ob:tree_probability_extended}, the probability that TRM locates the facility below $s$ in the tree is $(\treesize(q) - c(q)) / (n/2)$. If the facility is indeed below $s$, then $n - \overline{\treesize}(s)$ agents use $s$ under TRM. However, if the facility is above $s$, then $\overline{\treesize}(s)$ agents use $s$ under TRM. Hence, the expected number of agents using $s$ under TRM is:
\begin{align*}
    &
    \frac{\treesize(q)-c(q)}{n/2} \cdot \left(n - \overline{\treesize}(s)\right) + \left(1 - \frac{\treesize(q)-c(q)}{n/2}\right) \cdot \overline{\treesize}(s) \\
    ={} &
    \frac{\treesize(q)-c(q)}{n/2} \cdot \left(n - 2 \cdot \overline{\treesize}(s)\right) + \overline{\treesize}(s)\\
    ={} &
    \frac{(c(q) - n) + (n - \treesize(q))}{n/2} \cdot \left(2 \cdot \overline{\treesize}(s) - n\right) + \overline{\treesize}(s)\\
    \leq{} &
    \frac{(c(q) - n) + 2(n - \overline{\treesize}(s))}{n/2} \cdot \left(2 \cdot \overline{\treesize}(s) - n\right) + \overline{\treesize}(s)\\
    ={} &
    \frac{\left(5n/4 - 2 \cdot \overline{\treesize}(s)\right) \cdot \left(2 \cdot \overline{\treesize}(s) - n\right) + \overline{\treesize}(s) \cdot (n/2)}{n/2}\\
    ={} &
    \frac{5n \cdot \overline{\treesize}(s) - 4 \cdot (\overline{\treesize}(s))^2 - 5n^2/4}{n/2}\\
    ={} &
    10 \cdot \overline{\treesize}(s) - 8n^{-1} \cdot (\overline{\treesize}(s))^2 - 5n/2 \enspace.
\end{align*}

Let us now consider the second case, \ie, the case that $q \not \in X'$. Notice that in this case $\treesize(q) < n/4$, and therefore, no point of of $X'$ is located in the subtree rooted at $q$. Thus, the probability that the facility will be located below $s$ is $0$, and exactly $\overline{\treesize}(s)$ agents use $s$ under TRM. Moreover, in this case, by Lemma~\ref{le:twice_case_2}:
\[
    \overline{\treesize}(s)
    \leq
    n - \frac{n - \treesize(q)}{2}
    =
    \frac{n + \treesize(q)}{2}
    <
    \frac{5n}{8}
    \enspace.
    \qedhere
\]
\end{proof}

On the other hand, we can also characterize the number of agents using $s$ under the optimal solution.

\begin{lemma} \label{le:opt_pays_case_2}
The number of agents that use $s$ under the optimal solution is: $n - \overline{\treesize}(s)$.
\end{lemma}
\begin{proof}
The proof of this lemma is analogous to the proof of Lemma~\ref{le:opt_pays}
\end{proof}

We can now conclude the proof for the case $\overline{\treesize}(s) \geq n/2$.

\begin{corollary} \label{co:segment_ratio_case_2}
The number of agents using $s$ under TRM is at most twice the number of agents using $s$ under the optimal solution.
\end{corollary}
\begin{proof}
If $10 \cdot \overline{\treesize}(s) - 8n^{-1} \cdot (\overline{\treesize}(s))^2 - 5n/2 \leq \overline{\treesize}(s) < 5n/8$, then $n - \overline{\treesize}(s) > 3n/8$ agents use $s$ under the optimal solution, while at most $\overline{\treesize}(s) < 5n/8$ use it, in expectation, under TRM. The ratio between these values is $5/3 < 2$.

Otherwise, by Lemmata~\ref{le:TRM_pays_case_2} and \ref{le:opt_pays_case_2}, the ratio between the number of agents using $s$ under TRM and under the optimal solution is at most:
\begin{align*} \label{eq:ratio_TRM_case_2}
    &
    \frac{10 \cdot \overline{\treesize}(s) - 8n^{-1} \cdot (\overline{\treesize}(s))^2 - 5n/2}{n - \overline{\treesize}(s)}\\
    ={} &
    \frac{10 \cdot [n - (n -\overline{\treesize}(s))] - 8n^{-1} \cdot [n - (n - \overline{\treesize}(s))]^2 - 5n/2}{n - \overline{\treesize}(s)}\\
    ={} &
    \frac{6 \cdot (n -\overline{\treesize}(s)) -8n^{-1} \cdot (n -\overline{\treesize}(s))^2 - n/2}{n - \overline{\treesize}(s)}\\
    ={} &
    6 -8n^{-1} \cdot (n -\overline{\treesize}(s)) - \frac{n}{2(n - \overline{\treesize}(s))}
    \enspace.
\end{align*}
It can be easily checked that the last expression is maximized for $\overline{\treesize}(s) = 3n/4$, and the value of this maximum is:
\[
    6 - 8n^{-1} \cdot (n - 3n/4) - \frac{n}{2 \cdot (n - 3n/4)}
    =
    6 - 2 - 2
    =
    2
    \enspace.
    \qedhere
\]
\end{proof}

Combing the results we proved for the two cases, we get the competitive ratio of TRM.

\begin{corollary}
The Tree Randomized Mechanism (TRM) is $2$-competitive.
\end{corollary}
\begin{proof}
Corollaries~\ref{co:segment_ratio} and \ref{co:segment_ratio_case_2} prove that the expected number of agents using segment $s$ under TRM is always at most double the number of such agents under the optimal solution. Hence, the expected contribution of $s$ to the cost of TRM is always at most double its contribution to the cost of the optimal solution. The corollary now follows if we recall that $s$ was chosen as an arbitrary segment of $S$, and that the entire cost of the optimal solution and TRM is contributed by segments of $S$.
\end{proof}

\subsection{Proof: TRM is Two-Sided IC} \label{th:truthfulness_general}

In this subsection we prove that TRM is Two-Sided IC. Intuitively, the proof has three steps.
\begin{compactitem}
    \item Showing that when a mediator deviates the ``probability mass'' of the facility location follows the median of the locations reported by the mediator. More specifically, if mediator $d$ deviates and reports a set of agent locations whose median is $\hat{\ell}$ instead of the true median $\ell$, then the ``probability mass'' moves towards $\hat{\ell}$ inside the path connecting $\ell$ and $\hat{\ell}$. Outside this path, the probabilities of nodes to become the facility location is unaffected. This step consists of Lemma~\ref{le:short_deviation} and Corollary~\ref{co:long_deviation}.
    \item Proving (using the first step) that whenever an agent located at $t$ deviates and reports a location $\hat{t} \neq t$, then the median of the locations reported by its mediator can only go further away from $t$, and therefore, the ``probability mass'' of the facility location also goes away from $t$. This step is achieved by Lemma~\ref{le:TRM_truthful}.
    \item Proving (using the first step) that whenever a mediator deviates and reports a set of agent locations whose median is $\hat{\ell}$ instead of the true median $\ell$, then the ``probability mass'' is pushed away from $\ell$, and therefore, things become worse for the agents of the deviating mediator. This step is achieved by Lemma~\ref{le:TRM_mediators_truthful}.
\end{compactitem}

Recall that $V$ is the set of nodes of the tree produced by the first part of TRM. Let us begin with the first step.

\begin{lemma} \label{le:short_deviation}
Assume mediator $d_i$ reports one of two multi-sets $S_i$ and $\hat{S}_i$ of agent locations. Let $u$ and $\hat{u}$ be the medians of $S_i$ and $\hat{S}_i$ closest to $z$, respectively. Assume that the path between $u$ and $\hat{u}$ goes through no other point of $V$. Let $p(v), X, V$ and $\hat{p}(v), \hat{X}, \hat{V}$ denote the value of $p(v)$, the set $X$ and the set $V$ given that $d_i$ reports $S_i$ and $\hat{S}_i$, respectively.
Also, let $P = \{u, \hat{u}\}$, then:
\begin{compactitem}
    \item For a node $v \in V \setminus P$, $p(v) = \hat{p}(v)$.
    \item $p(\hat{u}) \leq \hat{p}(\hat{u})$.
    \item $p(u) + p(\hat{u}) = \hat{p}(u) + \hat{p}(\hat{u})$.
\end{compactitem}
\end{lemma}
\begin{proof}
For every node $v \in V \cup \{\hat{u}\}$, let $\treesize(v), \subtree(v), c(v)$ and $\widehat{\treesize}(v)$, $\widehat{\subtree}(v), \hat{c}(v)$ denote the values of $\treesize(v), \subtree(v)$ and $c(v)$ given that $d_i$ reports $S_i$ and $\hat{S}_i$, respectively. If $u \not \in \hat{V}$, then we define $\widehat{\treesize}(u), \widehat{\subtree}(u)$ and $\hat{c}(u)$ like in Subsection~\ref{ssc:approximation_general}. Similarly, if $\hat{u} \not \in V$, then we define $\treesize(\hat{u}), \subtree(\hat{u})$ and $c(\hat{u})$ like in Subsection~\ref{ssc:approximation_general}. In the rest of the proof we distinguish between two cases.

\paragraph{First case:}
In the first case reporting either one of the multi-sets $S_i$ and $\hat{S}_i$ induces the same root $r$. Let us begin by proving that $p(v) = \hat{p}(v)$ for every point $v \in V \setminus P$. Let $a$ be the lowest common ancestor of $u$ and $\hat{u}$. Clearly, $a \in P$ since no other point of $V$ appears on the path between $u$ and $\hat{u}$. The subtree of every point $v \in V \setminus P$ which is not an ancestor of $a$ is the same (including the $\size$ of the nodes) regardless of which multi-set $d_i$ reports ($S_i$ or $\hat{S}_i$), and therefore, $p(v) = \hat{p}(v)$ because $p(v)$ depends only on the subtree of $v$ and on the question whether $v$ is the root. Consider now a point $v \in V \setminus P$ which is an ancestor of $a$, and let $v'$ be the child of $v$ that is an ancestor of $a$ when $d_i$ reports $S_i$ (possibly $v' = a$). Similarly, let $\hat{v}'$ be the child of $v$ that is an ancestor of $a$ when $d_i$ reports $\hat{S}_i$ (possibly $\hat{v}' = a$ and/or $v' = \hat{v}'$). Observe that the declaration of $d_i$ affects only the subtree of $a$, and therefore, does not affect any child of $v$ except for $v'$ and $\hat{v}'$. Moreover, either $v' = \hat{v}'$ or they are both in $P$, and thus, $\treesize(v') = \treesize(\hat{v}')$ and $v' \in X \Leftrightarrow \hat{v}' \in X$. Since these are the only two properties of $v'$ and $\hat{v}'$ used by the formula of $p(v)$, we get, once again, $p(v) = \hat{p}(v)$.

By Lemma~\ref{le:tree_probability} the sum of the probabilities of all points of $V$ is $1$. Thus, the change in $p(\hat{u})$ following the deviation plus the change in $p(u)$ must add up to $0$. Formally,
\[
    [\hat{p}(\hat{u}) - p(\hat{u})] + [\hat{p}(u) - p(u)] = 0
    \Rightarrow
    p(u) + p(\hat{u}) = \hat{p}(\hat{u}) + \hat{p}(u)
    \enspace.
\]

By the definition of $X$, if $a \not \in X$, then $\subtree(a) \cap X = \varnothing$ for the two declarations of $d_i$ we consider. Since both nodes of $P$ are in the subtree of $a$, we get in this case that every point $v \in P$ has $p(v) = \hat{p}(v) = 0$. Hence, the interesting case is when $a \in X$, which we assume next.

Assume $\hat{u} = a$. Notice that this implies $\hat{u} \in X$. We have three subcases to consider based on the membership of $u$ in $X$ and $\hat{X}$.
\begin{compactitem}
    \item If $u \not \in X$, then $u \not \in \hat{X}$ because the median of $S_i$ is in $u$, and the median of $\hat{S}_i$ is outside of $u$'s subtree. Hence, $p(u) = 0 = \hat{p}(u)$, which implies $p(\hat{u}) = \hat{p}(\hat{u})$, since we already proved $p(u) + p(\hat{u}) = \hat{p}(\hat{u}) + \hat{p}(u)$.
    \item If $u \in X$ and $u \not \in \hat{X}$, then $\hat{p}(u) = 0 \leq p(u)$, which implies $p(\hat{u}) \leq \hat{p}(\hat{u})$ (again, since we already proved $p(u) + p(\hat{u}) = \hat{p}(\hat{u}) + \hat{p}(u)$).
    \item If $u \in X$ and $u \in \hat{X}$, then recall that $\hat{u}$ is not in the subtree of $u$ and $c(u) = \hat{c}(u)$ because $u \neq r$. Thus, by applying Lemma~\ref{le:tree_probability} twice (once for $p$ and once for $\hat{p}$):
    \begin{align*}
        \hat{p}(u)
        ={} &
        \frac{\widehat{\treesize}(u) - \hat{c}(u)}{n/2} - \sum_{\substack{v \in \widehat{\subtree}(u) \\ v \neq u}} \hat{p}(v)\\
        ={} &
        \frac{\treesize(u) - n_i - c(u)}{n/2} - \sum_{\substack{v \in \subtree(u) \\ v \neq u}} p(v)
        =
        p(u) - \frac{2n_i}{n}
        \enspace,
    \end{align*}
    which implies $p(\hat{u}) < \hat{p}(\hat{u})$, and the proof continues as before.
\end{compactitem}

The case $u = a$ is analogous, and therefore, we omit it.

\paragraph{Second case:}
In the second case reporting each one of the multi-sets $S_i$ and $\hat{S}_i$ induces a different root node $r$. By Observation~\ref{ob:point_moves}, the root $r$ must move from $u$ to $\hat{u}$. Therefore, the deviation does not affect the subtree of every point $v \in V \setminus P$. Since $p(v)$ depends only on the subtree of $v$ and on the question whether $v$ is the root, we get that $p(v) = \hat{p}(v)$ for every point $v \in V \setminus P$.

Using the same argument as in the first case, we get from Lemma~\ref{le:tree_probability} that $p(u) + p(\hat{u}) = \hat{p}(\hat{u}) + \hat{p}(u)$. If $\hat{u} \not \in X$ then $p(\hat{u}) = 0 \leq \hat{p}(\hat{u})$, and we are done. Similarly, if $u \not \in \hat{X}$ then $\hat{p}(u) = 0 \leq p(u)$, which implies $p(\hat{u}) \leq \hat{p}(\hat{u})$ since we already proved $p(u) + p(\hat{u}) = \hat{p}(\hat{u}) + \hat{p}(u)$. Thus, the interesting case is when both $\hat{u} \in X$ and $u \in \hat{X}$, which imply also $u \in X$ and $\hat{u} \in \hat{X}$.

Notice that $\subtree(\hat{u}) = V \setminus  \widehat{\subtree}(u)$ and $V = \widehat{\subtree}(\hat{u}) $, thus
$\subtree(\hat{u}) = \widehat{\subtree}(\hat{u}) \setminus  \widehat{\subtree}(u)$.
Moreover, $\treesize(\hat{u}) = n - \widehat{\treesize}(u) - n_i$. Combining both observations:
\begin{align*}
    p(\hat{u})
    ={} &
    \frac{\treesize(\hat{u}) - c(\hat{u})}{n/2} - \sum_{\substack{v \in \subtree(\hat{u}) \\ v \neq \hat{u}}} p(v)\\
    \stackrel{(*)}{=}{} &
    \frac{\treesize(\hat{u}) - c(\hat{u})}{n/2} - \sum_{\substack{v \in \widehat{\subtree}(\hat{u}) \\ v \neq \hat{u}}} \hat{p}(v) + \sum_{v \in \widehat{\subtree}(u)} \hat{p}(v)\\
    ={} &
    \frac{[n - \widehat{\treesize(u)} - n_i] - c(\hat{u})}{n/2} - \sum_{\substack{v \in \widehat{\subtree}(\hat{u}) \\ v \neq \hat{u}}} \hat{p}(v) + \frac{\widehat{\treesize}(u) - \hat{c}(u)}{n/2}\\
    ={} &
    \frac{n - n_i - c(\hat{u}) - \hat{c}(u)}{n/2} - \sum_{\substack{v \in \widehat{\subtree}(\hat{u}) \\ v \neq \hat{u}}} \hat{p}(v)\\
    \stackrel{(**)}{=}{} &
    \frac{\widehat{\treesize}(\hat{u}) - \hat{c}(\hat{u}) - n_i}{n/2} - \sum_{\substack{v \in \widehat{\subtree}(\hat{u}) \\ v \neq \hat{u}}} \hat{p}(v)
    =
    \hat{p}(\hat{u}) - \frac{2n_i}{n}
    \enspace,
\end{align*}
where (*) holds because $\subtree(\hat{u}) = \widehat{\subtree}(\hat{u}) \setminus  \widehat{\subtree}(u)$ and (**) holds $\widehat{\treesize}(\hat{u}) = n$, $c(\hat{u}) = \hat{c}(u) = n/4$ and $\hat{c}(\hat{u}) = n/2$. Hence, $p(\hat{u}) < \hat{p}(\hat{u})$.
\end{proof}

The next corollary generalizes the last lemma to the case where the two multi-sets $S_i$ and $\hat{S}_i$ that $d_i$ can declare have medians that are further apart.

\begin{corollary} \label{co:long_deviation}
Assume mediator $d_i$ reports one of two multi-sets $S_i$ and $\hat{S}_i$ of agent locations. Let $u$ and $\hat{u}$ be the medians of $S_i$ and $\hat{S}_i$, respectively.
Let $p(v)$ and $\hat{p}(v)$ denote the value of $p(v)$ given that $d_i$ reports $S_i$ and $\hat{S}_i$, respectively. We also define $P(u, \hat{u})$ to be the path between $u$ and $u_i$, and let $u = w_1, w_2, \ldots, w_s = \hat{u}$ denote the nodes along $P(u, \hat{u})$. Then,
\begin{compactitem}
    \item For a node $v \in V \setminus P(u, \hat{u})$, $p(v) = \hat{p}(v)$.
    \item For $1 \leq k \leq s$, $\sum_{j = k}^s p(w_j) \leq \sum_{j = k}^s \hat{p}(w_j)$.
\end{compactitem}
\end{corollary}
\begin{proof}
Let us consider $s$ multi-sets $S = S_1, S_2, \ldots, S_s = \hat{S}$ that $d_i$ can declare, where the median of $S_h$ is $w_h$, and let $p_h(v)$ denote the value of $p(v)$ assuming $d_i$ declares the multi-set $S_h$. Lemma~\ref{le:short_deviation} gives the following properties of the relation between $p_h$ and $p_{h + 1}$ for every $1 \leq h < s$.
\begin{compactitem}
    \item A point $v \in V \setminus P(u, \hat{u})$ is not the median of either $S_h$ or $S_{h+1}$, and therefore, $p_h(v) = p_{h + 1}(v)$.
    \item For $k > h + 1$: $p_h(w_j) = p_{h + 1}(w_j)$ for every $j \geq k$, and thus, $\sum_{j = k}^s p_h(w_j) = \sum_{j = k}^s p_{h + 1}(w_j)$.
    \item For $k = h + 1$: $p_h(w_k) \leq p_{h + 1}(w_k)$ and $p_h(w_j) = p_{h + 1}(w_j)$ for every $j > k$. Thus, $\sum_{j = k}^s p_h(w_j) \leq \sum_{j = k}^s p_{h + 1}(w_j)$.
    \item For $k \leq h$: $p_h(w_j) = p_{h + 1}(w_j)$ for every $j$ such that $k \leq j < h$ or $h + 1 < j$. Also, $p_h(w_i) + p_h(w_{h + 1}) = p_{h + 1}(w_h) + p_{h + 1}(w_{h + 1})$. Thus, $\sum_{j = k}^s p_h(w_j) = \sum_{j = k}^s p_{h + 1}(w_j)$.
\end{compactitem}

\vspace{0.1in}Combining the above properties, we can conclude:
\begin{compactitem}
    \item For a point $v \in V \setminus P(u, \hat{u})$, $p(v) = p_1(v) = p_2(v) = \ldots = p_s(v) = \hat{p}(v)$.
    \item For every $1 \leq k \leq s$, $\sum_{j = k}^s p(w_j) = \sum_{j = k}^s p_1(w_j) \leq \sum_{j = k}^s p_2(w_j) \leq \ldots \leq \sum_{j = k}^s p_s(w_j) = \sum_{j = k}^s \hat{p}(w_j)$.
    \qedhere
\end{compactitem}
\end{proof}

Next, we prove a simple technical observation that is used by the proofs of the following lemmata.
\begin{observation} \label{ob:sum_exchange}
Given a series $v_0, v_1, \ldots, v_m$ of points in the tree such that the path from $v_0$ to $v_m$ goes through $v_1, v_2, \ldots, v_{m - 1}$ in this order, and values $\{A_f, B_f ~|~ 1 \leq f \leq m\}$ such that: $\sum_{f = h}^m A_f \leq \sum_{f = h}^m B_f$ for every $1 \leq h \leq m$. Then:
\[
    \sum_{h = 1}^m \dist(v_0, v_h) \cdot A_h
    \leq
    \sum_{h = 1}^m \dist(v_0, v_h) \cdot B_h
    \enspace.
\]
\end{observation}
\begin{proof}
Observe that $\dist(v_0, v_h) = \sum_{f = 1}^h \dist(v_{f - 1}, v_f)$. Therefore,
\begin{align*}
    \sum_{h = 1}^m \dist(v_0, v_h) \cdot A_h
    ={} &
    \sum_{h = 1}^m \left(\dist(v_{h - 1}, v_h) \cdot \sum_{f = h}^m A_f\right)\\
    \leq{} &
    \sum_{h = 1}^m \left(\dist(v_{h - 1}, v_h) \cdot \sum_{f = h}^m B_f\right)
    =
    \sum_{h = 1}^m \dist(v_0, v_h) \cdot B_h
    \enspace.
    \qedhere
\end{align*}
\end{proof}

In the rest of this appendix, we use Corollary~\ref{co:long_deviation} and Observation~\ref{ob:sum_exchange} to prove that TRM is two-sided~IC.

\begin{lemma} \label{le:TRM_truthful}
The Tree Randomized Mechanism (TRM) is agent-side IC.
\end{lemma}
\begin{proof}
Consider an agent $a_{i, j}$, and assume $d_i$ is truthful. We would like to show that it not beneficial for $a_{i, j}$ to deviate and report a fake location $\hat{t}_{i, j}$ instead of its real location $t_{i, j}$. If the deviation does not affect the median of the agents' locations reported by $d_i$, then the deviation does not affect the outcome of TRM, and is clearly not helpful for $a_{i, j}$. Hence, the interesting case is when the median of the locations reported by $d_i$ changes following the deviation from $u$ to $\hat{u}$. By Observation~\ref{ob:point_moves}, $u$ and $\hat{u}$ are both located along the path connecting $t_{i, j}$ and $\hat{t}_{i, j}$, and $u$ is closer to $t_{i, j}$ than $\hat{u}$.

Let $P(u, \hat{u})$ denote the path connecting $u$ and $\hat{u}$, and let us denote the points of $V \cup \{\hat{u}\}$ along the path $P(u, \hat{u})$ by $u = w_1, w_2, \ldots, w_s = \hat{u}$. Also, let $p(v), V$ and $\hat{p}(v), \hat{V}$ denote the probability $p(v)$ and the set $V$ before and after the deviation, respectively.
We can now evaluate the cost paid by $a_{i, j}$:
\begin{align*}
    \sum_{v \in V} p(v) \cdot & \dist(v, t_{i, j})
    =
    \sum_{v \in V \setminus P(u, \hat{u})} p(v) \cdot \dist(v, t_{i, j}) + \sum_{h = 1}^s p(w_h) \cdot \dist(w_h, t_{i, j})\\
    \leq{} &
    \sum_{v \in V \setminus P(u, \hat{u})} \hat{p}(v) \cdot \dist(v, t_{i, j}) + \sum_{h = 1}^s \hat{p}(w_h) \cdot \dist(w_h, t_{i, j})
    =
    \sum_{v \in \hat{V}} \hat{p}(v) \cdot \dist(v, t_{i, j})
    \enspace,
\end{align*}
where the inequality follows from Observation~\ref{ob:sum_exchange} (when $v_0 = t_{i, j}$, $v_h = w_h$, $A_h = p(w_h)$ and $B_h = \hat{p}(w_h)$ for every $1 \leq h \leq s$) and Corollary~\ref{co:long_deviation}. Notice that the leftmost hand side is the cost of $a_{i, j}$ before the deviation, and the rightmost hand side is the cost of $a_{i, j}$ after the deviation. Hence, the above inequality proves that $a_{i, j}$ has no incentive to deviate. The lemma follows since $a_{i, j}$ is an arbitrary agent.
\end{proof}

\begin{lemma} \label{le:TRM_mediators_truthful}
The Tree Randomized Mechanism (TRM) is mediator-side IC.
\end{lemma}
\begin{proof}
Consider an arbitrary mediator $d_i$, and let us assume the agents of $d_i$ are truthful. Let $\ell_i$ be a median of the reports of all the agents represented by $d_i$, and let $\hat{\ell}_i$ be an arbitrary point. Our objective is to show that the expected social cost of the agents of $A_i$ does not improve if $d_i$ deviates and reports a set of agent locations whose median is $\hat{\ell}_i$.

Let $p(v)$ and $\hat{p}(v)$ denote the probability $p(v)$ before and after the deviation, respectively. Also, let $\ell_i = w_1, w_2, \ldots, w_s = \hat{\ell}_i$ denote the points of $V \cup \{\hat{\ell}_i\}$ along the path from $\ell_i$ to $\hat{\ell}_i$. To simplify the notation, let $p^\Delta(v) = \hat{p}(v) - p(v)$ be the change in $p(v)$ following the deviation. By Corollary~\ref{co:long_deviation}, $p^\Delta(v) = 0$ for every point $v \in V$ which is not on the path between $\ell_i$ and $\hat{\ell}_i$. Hence, the change in the expected social cost of the agents represented by $d_i$ following the deviation is:
\begin{equation} \label{eq:mediator_change_general}
    \sum_{j = 1}^{n_i} \sum_{h = 1}^s \left[\dist(t_{i,j}, w_h) \cdot p^\Delta(w_h)\right] \enspace.
\end{equation}
To prove the lemma, it is enough to show that the last expression is always non-negative. Let us partition $A_i$ into two kinds of agents: $A'_i$ contains agents for which the path connecting them to $\hat{\ell}_i$ goes through $\ell_i$, and $A''_i$ contains the rest of the agents represented by $d_i$. Notice that all the agents of $A''_i$ are connected to each other via paths that do not go through $\ell_i$. Hence, by the definition of a median, $|A''_i| \leq |A_i|/2$, which implies $|A'_i| \geq |A''_i|$. For every agent $a_{i, j} \in A'_i$ and point $w_h$, we have $\dist(t_{i, j}, w_h) = \dist(t_{i, j}, \ell_i) + \dist(\ell_i, w_h)$.
Hence,
\begin{align}
    &
    \sum_{a_{i, j} \in A'_i} \sum_{h = 1}^s \left[\dist(t_{i, j}, w_h) \cdot p^\Delta(w_h)\right] \nonumber\\
    ={} &
    \sum_{a_{i, j} \in A'_i} \sum_{h = 1}^s \left[\dist(t_{i, j}, \ell_i) \cdot p^\Delta(w_h)\right] + \sum_{a_{i, j} \in A'_i} \sum_{h = 1}^s \left[\dist(\ell_i, w_h) \cdot p^\Delta(w_h)\right] \nonumber\\
    ={}&
    \sum_{a_{i, j} \in A'_i} \left[\dist(t_{i, j}, \ell_i) \cdot \sum_{h = 1}^s p^\Delta(w_h)\right] + |A'_i| \cdot \sum_{h = 1}^s \left[\dist(\ell_i, w_h) \cdot p^\Delta(w_h)\right] \nonumber\\
    ={}&
    |A'_i| \cdot \sum_{h = 1}^s \left[\dist(\ell_i, w_h) \cdot p^\Delta(w_h)\right] \label{eq:s1_sum}
    \enspace,
\end{align}

where the last equality follows from Corollary~\ref{co:long_deviation} since:
\[
    \sum_{h = 1}^s p(w_i)
    =
    1 - \sum_{v \in V \setminus \{w_1, w_2, \ldots, w_s\}} \mspace{-36mu} p(v) \mspace{9mu}
    =
    1 - \sum_{v \in V \setminus \{w_1, w_2, \ldots, w_s\}} \mspace{-36mu} \hat{p}(v) \mspace{9mu}
    =
    \sum_{h = 1}^s \hat{p}(w_i)
		\Rightarrow
		\sum_{h = 1}^s p^\Delta(w_i) = 0
		\enspace.
\]
For every agent $a_{i, j} \in A''_i$, let $w_{h(i, j)}$ denote the first point from the list $w_1, w_2, \ldots, w_s$ that appears on the path from $t_{i, j}$ to $\ell_i$. Then,
\begin{align*}
    &
    \sum_{a_{i, j} \in A''_i} \sum_{h = 1}^s \left[\dist(t_{i, j}, w_h) \cdot p^\Delta(w_h)\right] \\
    ={} &
    \sum_{a_{i, j} \in A''_i} \sum_{h = 1}^{h(i, j)} \left[\dist(t_{i, j}, w_h) \cdot p^\Delta(w_h)\right] + \sum_{a_{i, j} \in A''_i} \sum_{h = h(i, j) + 1}^s \left[\dist(t_{i, j}, w_h) \cdot p^\Delta(w_h)\right] \\
    ={} &
    \sum_{a_{i, j} \in A''_i} \sum_{h = 1}^{h(i, j)} \left[\dist(t_{i, j}, w_h) \cdot p^\Delta(w_h)\right]    + \sum_{a_{i, j} \in A''_i} \left[\dist(t_{i, j}, w_{h(i, j) + 1}) \cdot \sum_{h = h(i, j) + 1}^s \mspace{-18mu} p^\Delta(w_h)\right]\\
    &+
    \sum_{a_{i, j} \in A''_i} \sum_{h = h(i, j) + 2}^s \mspace{-18mu} \left[\dist(w_{h(i, j) + 1}, w_h) \cdot p^\Delta(w_h)\right]\\
    \geq{} &
    \sum_{a_{i, j} \in A''_i} \sum_{h = 1}^{h(i, j)} \left[\dist(t_{i, j}, w_h) \cdot p^\Delta(w_h)\right]    + \sum_{a_{i, j} \in A''_i} \left[\dist(t_{i, j}, w_{h(i, j) + 1}) \cdot \sum_{h = h(i, j) + 1}^s \mspace{-18mu} p^\Delta(w_h)\right]
    \enspace,
\end{align*}

where the inequality follows from Observation~\ref{ob:sum_exchange} (when $v_0 = w_{h(i,j) + 1}$, $v_h = w_{h + 1 + h(i, j)}$, $A_h = 0$ and $B_h = \hat{p}(w_{h + 1 + h(i, j)}) - p(w_{h + 1 + h(i, j)})$ for every $1 \leq h \leq s - h(i, j) - 1$) and Corollary~\ref{co:long_deviation}.
Notice that for every $h \leq h(i, j)$: $\dist(t_{i, j}, w_h) = \dist(t_{i, j}, \ell_i) - \dist(\ell_i, w_h)$. Hence:
\begin{align*}
    &
    \sum_{a_{i, j} \in A''_i} \sum_{h = 1}^{h(i, j)} \left[\dist(t_{i, j}, w_h) \cdot p^\Delta(w_h)\right]\\
    ={} &
    \sum_{a_{i, j} \in A''_i} \left[\dist(t_{i, j}, \ell_i) \cdot \sum_{h = 1}^{h(i, j)}p^\Delta(w_h)\right] - \sum_{a_{i, j} \in A''_i} \sum_{h = 1}^{h(i, j)} \left[\dist(\ell_i, w_h) \cdot p^\Delta(w_h)\right]
    \enspace.
\end{align*}

Additionally, $\dist(t_{i, j}, w_{h(i, j) + 1}) \geq \dist(t_{i, j}, \ell_i) - \dist(\ell_i, w_{h(i, j) + 1})$, and $\sum_{h = h(i, j) + 1}^s p^\Delta(w_h) \geq 0$ by Corollary~\ref{co:long_deviation}. Hence:
\begin{align*}
    &
    \sum_{a_{i, j} \in A''_i} \left[\dist(t_{i, j}, w_{h(i, j) + 1}) \cdot \sum_{h = h(i, j) + 1}^s \mspace{-18mu} p^\Delta(w_h)\right] \\
    \geq{} &
    \sum_{a_{i, j} \in A''_i} \left[\dist(t_{i, j}, \ell_i) \cdot \sum_{h = h(i, j) + 1}^{s} \mspace{-18mu} p^\Delta(w_h)\right]
    - \sum_{a_{i, j} \in A''_i} \left[\dist(\ell_i, w_{h(i, j) + 1}) \cdot \sum_{h = h(i, j) + 1}^{s} \mspace{-18mu} p^\Delta(w_h)\right]\\
    ={} &
    \sum_{a_{i, j} \in A''_i} \left[\dist(t_{i, j}, \ell_i) \cdot \sum_{h = h(i, j) + 1}^{s} \mspace{-18mu} p^\Delta(w_h)\right]
    - \sum_{a_{i, j} \in A''_i} \sum_{h = h(i, j) + 1}^{s} \mspace{-18mu} \left[\dist(\ell_i, w_h) \cdot p^\Delta(w_h)\right]\\
    &+ \sum_{a_{i, j} \in A''_i} \sum_{h = h(i, j) + 2}^{s} \mspace{-18mu} \left[\dist(w_{h(i, j) + 1}, w_h) \cdot p^\Delta(w_h)\right] \\
    \geq{} &
    \sum_{a_{i, j} \in A''_i} \left[\dist(t_{i, j}, \ell_i) \cdot \sum_{h = h(i, j) + 1}^{s} \mspace{-18mu} p^\Delta(w_h)\right]
    - \sum_{a_{i, j} \in A''_i} \sum_{h = h(i, j) + 1}^{s} \mspace{-18mu} \left[\dist(\ell_i, w_h) \cdot p^\Delta(w_h)\right]
    \enspace,
\end{align*}
where the second inequality follows again from Observation~\ref{ob:sum_exchange} (when $v_0 = w_{h(i,j) + 1}$, $v_h = w_{h + 1 + h(i, j)}$, $A_h = 0$ and $B_h = \hat{p}(w_{h + 1 + h(i, j)}) - p(w_{h + 1 + h(i, j)})$ for every $1 \leq h \leq s - h(i, j) - 1$) and Corollary~\ref{co:long_deviation}. Combing the last results, we get:
\begin{align}
    &
    \sum_{a_{i, j} \in A''_i} \sum_{h = 1}^s \left[\dist(t_{i, j}, w_h) \cdot p^\Delta(w_h)\right] \nonumber\\
    \geq{} &
    \sum_{a_{i, j} \in A''_i} \left[\dist(t_{i, j}, \ell_i) \cdot \sum_{h = 1}^s p^\Delta(w_h)\right] - \sum_{a_{i, j} \in A''_i} \sum_{h = 1}^s \left[\dist(\ell_i, w_h) \cdot p^\Delta(w_h)\right]
    \nonumber\\
    ={}&
    - |A''_i| \cdot \sum_{h = 1}^s \left[\dist(\ell_i, w_h) \cdot p^\Delta(w_h)\right]  \label{eq:s2_sum}
    \enspace,
\end{align}
where the equality holds since $\sum_{h = 1}^{s}[\hat{p}(w_h) - p(w_h)] = 0$ by Corollary~\ref{co:long_deviation}. Adding up Expressions~(\ref{eq:s1_sum}) and (\ref{eq:s2_sum}) we get that (\ref{eq:mediator_change_general}) is lower bounded by:
\begin{align*}
    \sum_{j = 1}^{n_i} \sum_{h = 1}^s [\dist(t_{i,j}, w_h) \cdot &p^\Delta(w_h)]
    \geq
    (|A'_i| - |A''_i|) \cdot \sum_{h = 1}^s \left[\dist(\ell_i, w_h) \cdot p^\Delta(w_h)\right]\\
    ={} &
    (|A'_i| - |A''_i|) \cdot \sum_{h = 1}^{s - 1} \left[\dist(w_h, w_{h + 1}) \cdot \sum_{j = h + 1}^s p^\Delta(w_h)\right]
    \geq
    0
    \enspace,
\end{align*}
where the second inequality follows since the internal sum is non-negative due to Corollary~\ref{co:long_deviation} and $|A'_i| \geq |A''_i|$.
\end{proof}
}

\end{document}